\documentclass[prb,article]{revtex4-2}

\usepackage{amsmath}  
\usepackage{amsfonts} 
\usepackage{graphicx} 
\usepackage{hyperref}
\usepackage{amsthm}

\newcommand{\beq}{\begin{equation}}
	\newcommand{\eeq}{\end{equation}}
\newcommand{\bqa}{\begin{eqnarray}}
	\newcommand{\eqa}{\end{eqnarray}}
\newcommand{\nn}{\nonumber}

\newcommand{\smallfrac}[2]{\mbox{$\frac{#1}{#2}$}}

\newcommand{\half}{\smallfrac{1}{2}}

\newtheorem{Theorem}{Theorem}

\newtheorem{Corollary}{Corollary}
\newtheorem{Lemma}{Lemma}

\newcommand{\tr}{{\rm tr}}
\newcommand{\E}{{\cal E}}

\begin{document}
\title{Better Heisenberg limits, coherence bounds, and energy-time tradeoffs \\via quantum R\'enyi information}

\author{Michael J. W. Hall}
\affiliation{Theoretical Physics, Research School of Physics, Australian National University, Canberra ACT 0200, Australia}

\begin{abstract}An uncertainty relation for the R\'enyi entropies of conjugate quantum observables is used to obtain a strong Heisenberg limit of the form ${\rm RMSE}\geq f(\alpha)/(\langle N\rangle+\half)$,  bounding the root mean square error of any estimate of a random optical phase shift in terms of average photon number, where $f(\alpha)$ is maximised for non-Shannon entropies. Related simple yet strong uncertainty relations linking phase uncertainty to the photon number distribution, such as $\Delta\Phi\geq \max_n p_n$, are also obtained. These results are significantly strengthened via upper and lower bounds on the R\'enyi mutual information of quantum communication channels, related to asymmetry and convolution, and applied to the estimation (with prior information) of unitary shift parameters such as rotation angle and time, and to obtain strong bounds on measures of coherence. Sharper R\'enyi entropic uncertainty relations are also obtained, including time-energy uncertainty relations for Hamiltonians with discrete spectra. In the latter case almost-periodic R\'enyi entropies are introduced for nonperiodic systems.
\end{abstract}

\maketitle

\section{Introduction}
\label{sec:intro}

Quantum mechanics places fundamental limits on the information which can be gained in various contexts, ranging from the accuracy to which the phase shift of an optical probe state can be estimated to the secure key rate that can be obtained from a cryptographic protocol. 
Such limits are often formulated via uncertainty relations that restrict, for example, the degree to which values of two observables can be jointly specified, or the degree to which both an intended party and an eavesdropper can access quantum information~\cite{Colesreview}.

Entropic uncertainty relations place particularly strong restrictions, and underlie the main themes of this paper. One has, for example, the number-phase uncertainty relation~\cite{BBM, Frank2012}
\beq \label{ur}
H(N|\rho) + H(\Phi|\rho) \geq \log 2\pi + H(\rho),
\eeq
for the number and canonical phase observables of an optical mode, $N$ and $\Phi$. Here $H(A|\rho)=-\sum_a p(a|\rho)\log p(a|\rho)$ is the Shannon entropy of an observable $A$ with probability distribution $p(a|\rho)$, for a state described by density operator $\rho$, and $H(\rho)=-\tr[\rho\log \rho]$ denotes the von Neumann entropy of the state. The choice of logarithm base is left open throughout, corresponding to a choice of units, e.g., to bits for base 2 and nats for base $e$.  It follows that the number and phase uncertainties of any quantum state, as quantified by their Shannon entropies, cannot both be arbitrarily small. 

Entropic uncertainty relations have useful counterparts in quantum metrology. For example, if a random phase shift $\Theta$ of an optical probe state $\rho$ is estimated via some measurement $\Theta_{\rm est}$, then it follows from uncertainty relation~(\ref{ur}) that the error in the estimate, $\Theta_{\rm est}-\Theta$, is strongly constrained by the tradeoff relation~\cite{HallPRX}
\beq \label{urmet}
H(N|\rho) + H(\Theta_{\rm est}-\Theta|\rho) \geq \log 2\pi + H(\rho) .
\eeq
This relation applies to arbitrary estimates, rather than to a particular phase observable $\Phi$,  and further implies that the root-mean-square error (RMSE) of the estimate is bounded by~\cite{HallPRX,HallPRA,Nair}
\beq \label{rmse}
{\rm RMSE} :=\langle (\Theta_{\rm est}-\Theta)^2\rangle^{1/2} \geq \frac{\sqrt{2\pi/e^3}}{\langle N\rangle +1} ,
	\eeq
where $\langle N\rangle=\tr[\rho N]$ denotes the average photon number of the probe state. This is a strong form of the well-known Heisenberg limit in quantum metrology, which states that the phase error can asymptotically scale no better than $\langle N\rangle^{-1}$,  where such limits cannot be obtained via quantum Fisher information methods without additional assumptions~\cite{ZZ} (see also Section~\ref{sec:setting}). The above bounds are valid for both linear and nonlinear phase shifts, can be further strengthened to take into account any prior information about the phase shift $\Theta$, and in many cases far outperform bounds based on quantum Fisher information~\cite{HallPRX} (see also Section~\ref{sec:ent}).

While the above results arise via properties of standard Shannon and von Neumann entropies, it is known that various elements of quantum information theory can be generalised to the family of R\'enyi entropies, $H_\alpha(A|\rho)$ and $H_\alpha(\rho)$, and associated R\'enyi relative entropies $D_\alpha(\rho\|\sigma)$~\cite{Colesreview}. These quantities are labelled by a real index, $\alpha\geq0$, and reduce to the standard entropies and relative entropy for $\alpha=1$. One has, for example, the R\'enyi uncertainty relation~\cite{Maassen, BB2006}
\beq \label{urrenyi}
H_\alpha(N|\rho) + H_\beta(\Phi|\rho) \geq \log 2\pi , \qquad \frac{1}{\alpha}+\frac{1}{\beta}=2,
\eeq
analogous to Equation~(\ref{ur}). Several questions then immediately arise. Are such generalisations to R\'enyi entropies advantageous? Why are the uncertainties of $N$ and $\Phi$ characterised by two different R\'enyi entropies, $H_\alpha$ and $H_\beta$, in Equation~(\ref{urrenyi})? And why is there no term depending on the degree of purity of the state, analogous to $H(\rho)$ in Equation~(\ref{ur})?

Several positive answers to the first question above are known, in contexts such as mutually unbiased bases~\cite{Maassen}, quantum cryptography~\cite{RennerQKD}, and quantum steering~\cite{Brunner}. An aim of this paper is to demonstrate further unambiguous advantages of R\'enyi entropic uncertainty relation~(\ref{urrenyi}), in the context of quantum metrology. For example, it will be shown in Section~\ref{sec:ent} to lead to a generalised  Heisenberg limit of the form
\beq \label{rmsealpha}
{\rm RMSE} \geq \frac{f(\alpha)}{\langle N\rangle +\half} 
\eeq
for random phase shifts, where the function $f(\alpha)$ is maximised for the choice $\alpha\approx 0.7471$. This choice not only improves on the denominator in Equation~(\ref{rmse}) (corresponding to $\alpha=1$), but also improves on the numerator, by around $4\%$, with the result being independent of R\'enyi entropies and any interpretation thereof. Further entropic bounds on the RMSE are obtained in Section~\ref{sec:ent}, as well as related simple yet strong uncertainty relations for number and canonical phase observables, such as
\beq \label{simple}
\Delta \Phi \geq \max_n p(n|\rho) .
\eeq

A second aim of the paper is to further strengthen uncertainty relations and metrology bounds such as Equations~(\ref{urmet})--(\ref{simple}), achieved in Section~\ref{sec:sand} via finding upper and lower bounds for the classical R\'enyi mutual information of quantum communication channels~\cite{Sibson,Verdu,Wilde,Tom}, which also shed light on the second and third questions above. The upper bounds are based on the notion of R\'enyi asymmetry~\cite{renyiasymm}, recently applied to energy-time uncertainty relations for conditional R\'enyi entropies by Coles {\it et al.}~\cite{Colestime}. The lower bounds relate to the convolution of the prior and error distributions. For example, the number-phase uncertainty relation
\beq \label{preview}
A_\alpha^N(\rho) + H_\alpha(\Phi|\rho)\geq \log 2\pi, \qquad \alpha\geq\half,
\eeq
is obtained in Section~\ref{sec:sand}, which generalises Equation~(\ref{ur}) for Shannon entropies, and strengthens Equation~(\ref{urrenyi}) for Renyi entropies to take the degree of purity of the state into account. Here $A^N_\alpha(\rho)$ denotes the associated R\'enyi asymmetry, which may be interpreted as quantifying the intrinsically `quantum' uncertainty of $N$, and satisfies a duality property for pure states that underpins the relationship between the indexes $\alpha$ and $\beta$ in Equation~(\ref{urrenyi}). 

The results in Section~\ref{sec:sand} hold for the general case of unitary displacements generated by an operator with a discrete spectrum (such as $N$). Applications to strong upper and lower bounds for several measures of coherence~\cite{coh,Chitambar}, the estimation of rotation angles, and energy-time metrology and uncertainty relations, are briefly discussed in Section~\ref{sec:time}. In the latter case, almost-periodic R\'enyi entropies are introduced for the time uncertainties of non-periodic systems, analogously to the case of standard entropies~\cite{Halltime}. Conclusions are given in Section~\ref{sec:con}, and  proof technicalities  are largely deferred to appendices.

\section{Metrology bounds, Heisenberg limit and uncertainty relations via R\'enyi entropies}
\label{sec:ent}

In this section an analogue of metrology relation~(\ref{urmet}) is derived for R\'enyi entropies, via uncertainty relation~(\ref{urrenyi}). The improved Heisenberg limit~(\ref{rmsealpha}) follows as a consequence, as well as several simple uncertainty relations for number and phase, including Equation~(\ref{simple}).  Stronger versions of these results will be obtained in Section~\ref{sec:sand}. 

\subsection{Definition of R\'enyi entropies and R\'enyi lengths}
\label{sec:length}

To proceed, several definitions are necessary.
First, the photon number  of an optical mode is described by a Hermitian operator $N$ having eigenstates $\{|n\rangle\}$, $n=0,1,2,\dots$, with associated probability distribution $p(n|\rho)=\langle n|\rho|n\rangle$ for a state described by density operator $\rho$. A phase shift $\theta$ of the field is correspondingly described by $\rho_\theta=e^{-iN\theta}\rho e^{iN\theta}$. 

Second, the canonically conjugate phase observable $\Phi$ is described by the positive-operator-valued measure (POVM) $\{|\phi\rangle\langle\phi|\}$, with $\phi$ ranging over the unit circle and
\beq \label{phi}
|\phi\rangle := \frac{1}{\sqrt{2\pi}} \sum_{n=0}^\infty e^{-in\phi} |n\rangle ,
\eeq
and associated canonical phase probability density  $p(\phi|\rho)=\langle\phi|\rho|\phi\rangle$~\cite{Helstrom,Holevo}. It is straightforward to check that this density is translated under phase shifts, i.e., $p(\phi|\rho_\theta)=p(\phi-\theta|\rho)$.

Third, the classical R\'enyi entropies of $N$ and $\Phi$ are defined by~\cite{Colesreview}
\beq \label{renyient}
H_\alpha(N|\rho) :=\frac{1}{1-\alpha}\log \sum_{n=0}^\infty p(n|\rho)^\alpha,\qquad 
H_\alpha(\Phi|\rho) :=\frac{1}{1-\alpha}\log \oint d\phi \, p(\phi|\rho)^\alpha ,
\eeq
for $\alpha\in[0,\infty)$. These reduce to the standard Shannon entropies in the limit $\alpha\rightarrow1$ (using, e.g., $\lim_{\alpha\rightarrow1}[g(\alpha)-g(1)]/[\alpha-1]=g'(1)$ for $g(\alpha)=\log \sum_n p(n|\rho)^\alpha$). They provide measures of uncertainty that are small for highly peaked distributions and large for spread-out distributions. In particular, $H_\alpha(N)=0$ and $H_\alpha(\Phi|\rho)=\log 2\pi$ for any number state $\rho=|n\rangle\langle n|$.  Direct measures of uncertainty are given by the associated R\'enyi lengths
\beq \label{length}
L_\alpha(N|\rho):=\left[\sum_{n=0}^\infty p(n|\rho)^\alpha\right]^\frac{1}{1-\alpha},\qquad 
L_\alpha(\Phi|\rho) :=\left[\oint d\phi \, p(\phi|\rho)^\alpha \right]^\frac{1}{1-\alpha} ,
\eeq
which  quantify the effective spreads of $N$ and $\Phi$ over the nonnegative integers and the unit circle, respectively~\cite{Hallvol}. Note that uncertainty relation~(\ref{urrenyi}) can be rewritten in the form
\beq
L_\alpha(N|\rho) L_\beta(\Phi|\rho) \geq 2\pi, \qquad \frac{1}{\alpha}+\frac{1}{\beta}=2
\eeq
for these spreads, akin to the usual Heisenberg uncertainty relation.

\subsection{Entropic tradeoff relation for phase estimation}

If some estimate $\theta_{\rm est}$ is made of a phase shift $\theta$ applied to a probe state, then the estimation error, $\theta_{\rm err}=\theta_{\rm est}-\theta$, will have a highly-peaked probability density for a good estimate, and a spread-out probability density for a poor estimate. Hence, the quality of the estimate can be quantified in terms of the R\'enyi entropy of $p(\theta_{\rm err})$. The following theorem imposes a tradeoff between the quality of any estimate of a completely unknown phase shift and the number entropy of the probe state.

\begin{Theorem} \label{thm1} For any estimate $\Theta_{\rm est}$ of a uniformly random phase shift $\Theta$ applied to a probe state $\rho$, the estimation error $\Theta_{\rm est}-\Theta$ satisfies the tradeoff relation
\beq \label{thm1b}
H_\alpha(\Theta_{\rm est}-\Theta|\rho) + H_\beta(N|\rho) \geq \log 2\pi,\qquad 
\frac{1}{\alpha}+\frac{1}{\beta}=2 .
\eeq
\end{Theorem}

Note that the condition on $\alpha$ and $\beta$ implies $\alpha,\beta\geq\half$. For the case of Shannon entropies, i.e., $\alpha=\beta=1$, this result has been previously obtained via entropic  uncertainty relation~(\ref{ur})~\cite{HallPRA}. A similar method is used in Appendix~\ref{appa} to prove the general result of the theorem via entropic uncertainty relation~(\ref{urrenyi}). 
It is worth emphasising that, unlike uncertainty relation~(\ref{urrenyi}), Theorem~\ref{thm1} applies to {\it any}  estimate of the random phase shift, including the canonical phase measurement $\Phi$ as a special case (for this case Equation~(\ref{urrenyi}) is recovered). 

Theorem~\ref{thm1} implies that  no phase-shift information can be gained via a probe state $|n\rangle\langle n|$, as expected since number eigenstates are insensitive to phase shifts. In particular, the number entropy $H_\beta(N|\rho)$  vanishes for any index $\beta$ and so the error entropy in Equation~(\ref{thm1b}) must reach its maximum value of $\log 2\pi$, which is only possible if the error has a uniform probability density, i.e.,  $p(\theta_{\rm err})=1/(2\pi)$.  

Conversely, Theorem~\ref{thm1} connects informative estimates with probe states that have a high number entropy. For example, if a canonical phase measurement is used to estimate a random phase shift of the pure probe state $|\psi\rangle=(n_{\max}+1)^{-1/2}(|0\rangle+|1\rangle+\dots|n_{\max}\rangle)$,   the error distribution may be calculated, using Equation~(\ref{perr}) of Appendix~\ref{appa} for $\Theta_{\rm est}\equiv\Phi$, as
\beq \label{opt}
p(\theta_{\rm err})=\langle \theta_{\rm err}|\rho|\theta_{\rm err}\rangle =|\langle\theta_{\rm err}|\psi\rangle|^2 =\frac{1}{2\pi (n_{\max}+1)} \left|\sum_{n=0}^{n_{\max}} e^{in\theta_{\rm err}}\right|^2 ,
\eeq 
and the Shannon entropy of the error then follows via Equation~(69) of Reference~\cite{HallJMO} as
\begin{align} \label{example}
H(\Theta_{\rm est}-\Theta|\rho) &= \log 2\pi + \log (n_{\max}+1) + 2[1 - 1^{-1} - 2^{-1} - \dots -(n_{\max}+1)^{-1}]\log e \nn\\
 &\leq \log 2\pi  - \log (n_{\max}+1) + 2(1-\gamma)\log e,
\end{align}
where $\gamma\approx0.5772$ is Euler's constant.  Hence such a probe state leads to an arbitrarily low uncertainty for the error as $n_{\max}$ increases. Moreover, Theorem~\ref{thm1} implies that this estimate is near-optimal, under the constraint of at most $n_{\max}$ photons, in the sense that the error entropy is within $2(1-\gamma)\log e\approx 1.2$ bits of the minimum possible, $\log 2\pi - \log (n_{\max}+1)$, allowed  by Equation~(\ref{thm1b}) under this constraint.

This last result strongly contrasts with Fisher information methods, which suggest that the best possible single-mode probe state, under the constraint of at most $n_{\max}$ photons, is the simple superposition state $2^{-1/2}(|0\rangle+|n_{\max}\rangle)$~\cite{CavesPRL}.  However,  it follows from Theorem~\ref{thm1} that this probe state cannot be optimal for the estimation of a random phase shift. In particular, noting that $H_\beta(N|\rho)=\log 2$ for this case, Equations~(\ref{length}) and~(\ref{thm1b}) give
\beq \label{fisher}
\log L_\alpha(\Theta_{\rm est}-\Theta|\rho) = H_\alpha(\Theta_{\rm est}-\Theta|\rho) \geq \log \pi
\eeq
for any value of $n_{\max}$, in stark contrast to Equation~(\ref{example}). Indeed, choosing  $\alpha=2$ gives 
\beq
\oint d\theta_{\rm err}\,\left[p(\theta_{\rm err})-\frac{1}{2\pi}\right]^2 = L_2(\Theta_{\rm est}-\Theta|\rho)^{-1} - \frac{1}{\pi} + \frac{1}{4\pi^2} \leq \frac{1}{4\pi^2},
\eeq
implying that $p(\theta_{\rm err})$ cannot be too different from a uniform distribution. Hence the simple superposition state has a poor performance in comparison to the probe state in Equation~(\ref{example}), for the case of uniformly random estimates.  A more  direct comparison with Fisher information bounds is made in in the following subsection, and the difference explained in Section~\ref{sec:setting}.

\subsection{Lower bounds for RMSE, a strong Heisenberg limit, and number-phase uncertainty relations }
\label{sec:rmse}

Equation~(\ref{length}) and Theorem~\ref{thm1} imply that the R\'enyi length of the error for any estimate of a random phase shift has the lower bound $L_\alpha(\Theta_{\rm est}-\Theta|\rho)\geq 2\pi/L_\beta(N|\rho)$. However, a more familiar length measure for characterising the performance of an estimation scheme is the root-mean-square error (RMSE) of the estimate, given by ${\rm RMSE} =\langle (\Theta_{\rm est}-\Theta)^2\rangle^{1/2}$.
Note that, in contrast to the case of entropies and R\'enyi lengths, a well-known ambiguity arises: $\theta_{\rm err}^2=(\theta_{\rm est}-\theta)^2$ is not a periodic function, and hence evaluation of the RMSE depends on the choice of a phase reference interval for the error $\theta_{\rm err}$. Fortunately this is easily resolved: a perfect estimate corresponds to a zero error, and hence the reference interval centred on zero, i.e., $\theta_{\rm err}\in[-\pi,\pi)$, will be used.

The following theorem gives three strong lower bounds for the RMSE, where the third has the form of a generalised Heisenberg limit as discussed in Section~\ref{sec:intro}. A corollary to this theorem, further below, gives corresponding preparation uncertainty relations for number and phase.

\begin{Theorem} \label{thm2}
For any estimate $\Theta_{\rm est}$ of a uniformly random phase shift $\Theta$ applied to a probe state $\rho$, the root-mean-square error RMSE=$\langle (\Theta_{\rm est}-\Theta)^2\rangle^{1/2}$ has the lower bounds
\beq \label{thm2a}
{\rm RMSE} \geq \frac{\pi}{\sqrt{3}\,L_{1/2}(N|\rho)},\qquad {\rm RMSE}\geq \max_n p(n|\rho), \qquad {\rm RMSE}\geq \frac{f_{\max}}{\langle N\rangle +\half} ,
\eeq
where $L_{1/2}(N|\rho)=\big[\sum_n \sqrt{p(n|\rho)}\,\big]^2$ is a R\'enyi length as defined in Equation~(\ref{length}), and $f_{\max}\approx 0.5823$ denotes the maximum value of the function 
\beq \label{thm2b}
f(\alpha):= \left\{ \begin{array}{ll}
	2\alpha^{-1} \left(\frac{\pi}{3\alpha-1}\right)^{\frac{1}{2}} \left(\frac{3\alpha-1}{2}\right)^{\frac{1}{1-\alpha}}
		\,	\frac{(1-\alpha)^{\frac{1}{2}}\Gamma(\frac{1}{1-\alpha})}{\Gamma(\frac{1}{1-\alpha}-\frac{1}{2})}, & \half\leq\alpha\leq1,\\
	2\alpha^{-1} \left(\frac{\pi}{3\alpha-1}\right)^{\frac{1}{2}} \left(\frac{3\alpha-1}{2}\right)^{\frac{1}{1-\alpha}}	
	\,
	\frac{(\alpha-1)^{\frac{1}{2}}\Gamma(\frac{\alpha}{\alpha-1}+\frac{1}{2})}{\Gamma(\frac{\alpha}{\alpha-1})}, & \alpha\geq 1,
	\end{array} \right.
\eeq
which is achieved for the choice $\alpha\approx 0.7471$.
\end{Theorem}
In Equation~(\ref{thm2b}), $\Gamma(x)$ denotes the Gamma function and the value of $f(1)$ is defined by taking the limit $\alpha\rightarrow1$ in either expression and using  $\lim_{x\rightarrow 0} (1-3x/2)^{1/x}=e^{-3/2}$ and $\lim_{x\rightarrow\infty}x^{1/2}\Gamma(x)/\Gamma(x+\half)=1$, to obtain $f(1)=\sqrt{2\pi/e^3}\approx 0.5593$.
 The scaling function~$f(\alpha)$ is plotted in Figure~\ref{fig1}. The proof of the theorem relies on Theorem~\ref{thm1} and upper bounds on R\'enyi entropies under various constraints, and is given in Appendix~\ref{appb}.

\begin{figure}[t!]
	\includegraphics[width=10.5 cm]{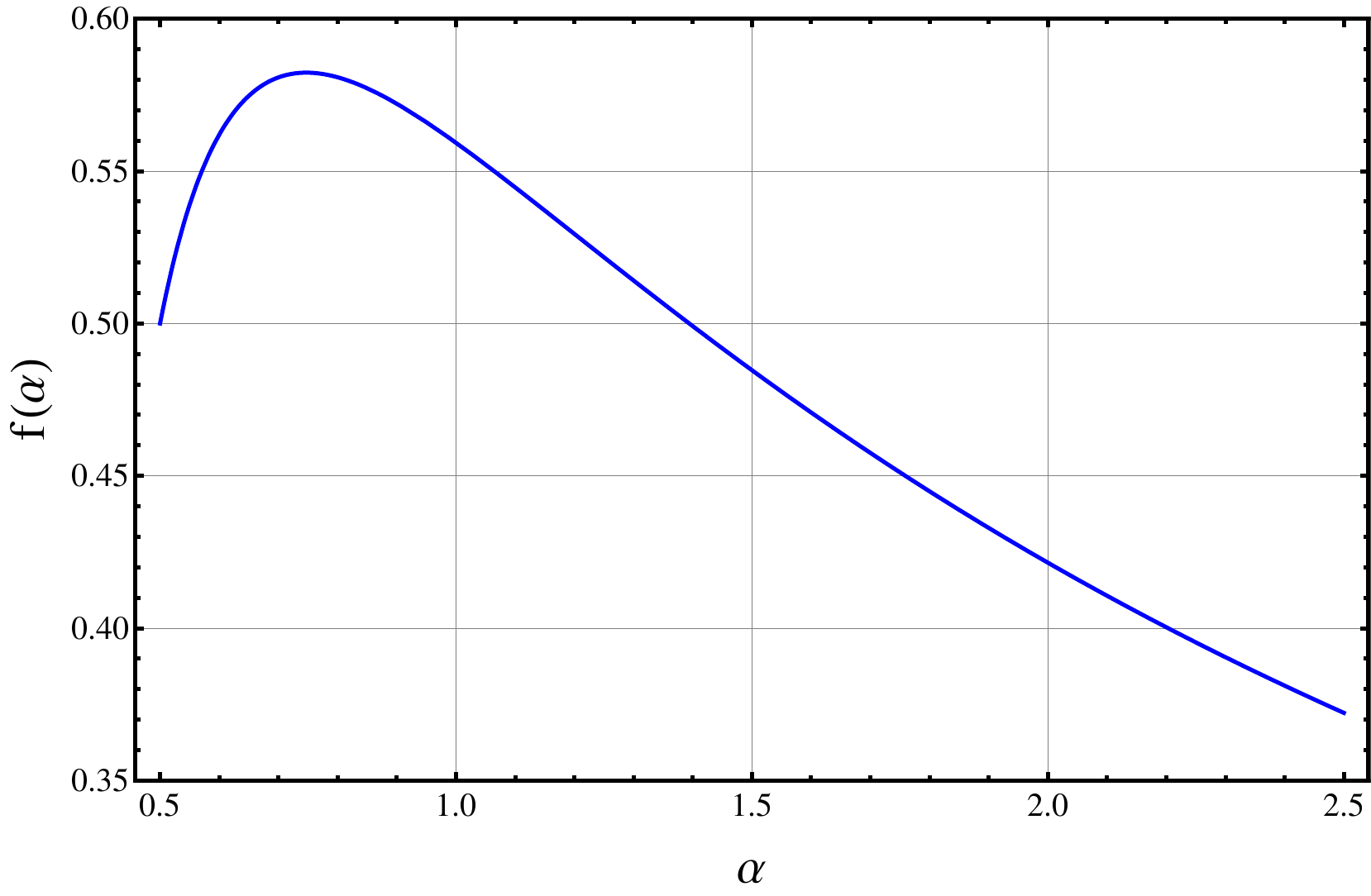}
	\caption{The scaling function $f(\alpha)$ for the Heisenberg limit in Theorem~\ref{thm2}. Particular values of interest are $f(1/2)=1/2$, $f(1)=\sqrt{2\pi/e^3}\approx 0.5593$, and the maximum value $f_{\max}\approx f(0.7471)\approx 0.5823$.
		\label{fig1}}
\end{figure}

The lower bounds in Theorem~\ref{thm2} are relatively strong, and indeed the first inequality in Equation~(\ref{thm2a}) is tight, being saturated for number states. In particular, the error distribution is always uniform for this case, as noted below Theorem~\ref{thm1}, yielding ${\rm RMSE}=\frac{1}{2\pi}\int_{-\pi}^\pi d\phi_{\rm err} \, (\phi_{\rm err})^2=\pi/\sqrt{3}$, as per the first lower bound.  

Moreover, the third bound in Equation~(\ref{thm2a}) of Theorem~\ref{thm2} is stronger than the Heisenberg limit in Equation~(\ref{rmse}), both in the numerator and the denominator. In particular, the scaling factor $f(1)\approx0.5593$ in Equation~(\ref{rmse}) is outperformed by $\approx 4\%$ when compared to $f_{\max}\approx 0.5823$ in Equation~(\ref{thm2a}).  Note that while the derivation of this bound relies on properties of R\'enyi entropies (see Appendix~\ref{appb}), no R\'enyi entropies appear in the bound itself. The bound thus demonstrates an unambiguous advantage of using R\'enyi entropies in quantum metrology that is completely independent of their interpretation.

Theorem~\ref{thm2} can also be directly compared to the Fisher information bound
\beq
{\rm RMSE}_\theta := \langle (\Theta_{\rm est}-\theta)^2\rangle_{\rho_\theta}^{1/2}=\left[ \int_{-\infty}^{\infty}d\theta_{\rm est}\ (\theta_{\rm est}-\theta)^2 p(\theta_{\rm est}|\rho_\theta)\right]^{1/2} \geq \frac{1}{2\Delta N} ,
\eeq
for the root mean square error of any locally unbiased estimate $\Theta_{\rm est}$ of a given phase-shift $\theta$ of probe state $\rho$~\cite{CavesPRL}. Here phase shifts are `unwrapped' from the unit circle to the real line, $\Delta N$ is the root mean square deviation of the number operator for the probe state, and local unbiasedness is the requirement that $\langle \Theta_{\rm est}\rangle_{\rho_\chi}=\int_{-\infty}^{\infty} d\theta_{\rm est}\,\theta_{\rm est}p(\theta_{\rm est}|\rho_\chi)=\chi$ for all phase shifts $\chi$ in some neighbourhood of $\theta$. Note the bound implies that ${\rm RMSE}_\theta$ becomes infinite for number states. Under the constraint of a maximum photon number $n_{\max}$, the probe state minimising the Fisher bound is the simple superposition $2^{-1/2}(|0\rangle+|n_{\max}\rangle)$, considered in Section~{\ref{sec:ent}, yielding ${\rm RMSE}_\theta \geq 1/n_{\max}$, which approaches zero as $n_{\max}$ increases~\cite{CavesPRL}. In contrast, for any estimate of a uniformly random phase shift, the first two bounds in Equation~(\ref{thm2a}) of Theorem~\ref{thm2} give the much stronger lower bounds ${\rm RMSE}>\pi/(2\sqrt{3})\approx 0.9069$ and ${\rm RMSE}\geq \half$ for this probe state, irrespective of the value of $n_{\max}$.  Thus the optimal single-mode probe state for the Fisher information bound is not optimal for estimating uniformly random phase shifts. The underlying reason for this difference is related to the degree of prior information available about the phase shift~\cite{HallPRX,HallJapan}, as discussed further in Section~\ref{sec:setting}.

Finally, it is of interest to note that each of the bounds in Theorem~\ref{thm2} can be used to obtain a corresponding preparation uncertainty relation for the canonical phase and photon number observables. In particular, defining the standard deviation $\Delta_\chi\Phi$ of the canonical phase observable $\Phi$ with respect to reference angle $\chi$ via~\cite{HallJMO}
\beq \label{phasevar}
(\Delta_\chi\Phi)^2 := \int_{\chi-\pi}^{\chi+\pi}  d\phi\, (\phi-\chi)^2 p(\phi|\rho),
\eeq
one has the following corollary of Theorem~\ref{thm2}, also proved in Appendix~\ref{appb}.

\begin{Corollary}
\label{cor1}
The canonical phase and photon number of an optical mode satisfy the family of uncertainty relations
\beq \label{cor1a}
L_\beta(N|\rho) \Delta_\chi\Phi \geq \alpha^{\frac{\alpha}{\alpha-1}} f(\alpha),\qquad \frac{1}{\alpha}+\frac{1}{\beta} = 2 ,
\eeq
for  R\'enyi length and standard deviation. In particular, the choices  $\alpha\rightarrow\infty$, $\alpha=0.5$, and the value of $\alpha$ maximising $f_{\alpha}$ yield the corresponding mumber-phase uncertainty relations
\beq \label{cor1b}
 L_{1/2}(N|\rho)\Delta_\chi\Phi \geq \frac{\pi}{\sqrt{3}},\qquad \Delta_\chi\Phi\geq \max_n p(n|\rho), \qquad  (\langle N\rangle +\half) \Delta_\chi\Phi\geq f_{\max}\approx 0.5823 .
\eeq
\end{Corollary}

The above uncertainty relations are easy to evaluate for many states, and are relatively strong. In particular, the first inequality in Equation~(\ref{cor1b}) is saturated for number states, and hence has the best possible lower bound. Further, the lower bound $f_{\max}$ in the third inequality is near-optimal, as it cannot be improved to more  than $\pi/(2\sqrt{3})$, corresponding to the value of the left hand side for the number state $|0\rangle$. This suggests the conjecture 
\beq \label{conj}
(\langle N\rangle +\half) \Delta_\chi\Phi\geq \frac{\pi}{2\sqrt{3}}\approx 0.9069 .
\eeq 
Note this conjecture is consistent with the fact that  the right hand side  can be no larger than $2(-z_A/3)^{3/2}\approx1.3761$ in the asymptotic limit $\langle N\rangle\rightarrow\infty$, where $z_A$ denotes the first (negative) zero
of the Airy function~\cite{Bandilla}.  Evidence for a related conjecture,  $(\langle N\rangle +1) \Delta_\chi\Phi\geq 2(-z_A/3)^{3/2}$, is given in~References~\cite{HallPRA,heislong}.

\section{Stronger metrology bounds and uncertainty relations via sandwiched R\'enyi relative entropies and asymmetry}
\label{sec:sand}

The results of the previous section relied on the known uncertainty relation~(\ref{urrenyi}) for Renyi entropies. To obtain stronger results which can, for example, take prior information and nonpurity into account, two approaches are possible. The first is to follow essentially the same strategy as in Section~\ref{sec:ent}, but starting with stronger uncertainty relations, similar to the approach in Reference~\cite{HallNJP} for standard entropies. The second approach, followed in this paper, is more fundamental, being based on information properties of quantum communication channels that not only yield stronger metrology tradeoff relations, but also lead to stronger uncertainty relations for R\'enyi entropies than Equation~(\ref{urrenyi}). Further, the results are applicable not only  to number and phase but to all unitary displacements generated by discrete-valued operators, including rotations generated by angular momentum and time evolution generated by a Hamiltonian with a discrete specturm.

\subsection{Setting the scene: the case of standard entropies}
\label{sec:setting}

It is helpful to first  briefly review how uncertainty relations and metrology bounds for standard Shannon and von Neumann entropies, such as Equations~(\ref{ur})--(\ref{rmse}), follow from upper and lower bounds for quantum information~\cite{HallPRX,Nair}. This both motivates and provides a base of comparison for the general case. 

For a quantum communication channel in which signal state $\rho_x$ is transmitted with prior probability density $p(x)$, corresponding to the ensemble $\E \equiv\{\rho_x;p(x)\}$,  the information that can be gained per signal via measurements of observable $A$ at the receiver is given by the Shannon mutual information
\beq \label{iax}
I(A:X) = H(A)-H(A|X)= H(A|\rho_\E) - \int dx\,p(x) H(A|\rho_x) , \qquad \rho_\E=\int dx\,p(x) \rho_x, 
\eeq
where $H(A|X)$ is the conditional entropy $H(AX)-H(X)$, and integration is replaced here and below by summation for discrete signal ensembles~\cite{inftext}. A useful upper bound for this information gain is the Holevo quantity $\chi(\E)$, with~~\cite{inftext,Ozawa}
\beq \label{hol}
I(A:X)\leq \chi(\E) := H(\rho_\E) - \int dx\,p(x) H(\rho_x) .
\eeq

Consider now the case of signal states generated by a group of unitary transformations $U_x=e^{-ixG}$, for some Hermitian generator $G$ with a discrete spectral decomposition $G=\sum_k g_k P_k$, where $P_k$ is the projector onto the eigenspace of eigenvalue $g_k$, so that $\rho_x=U_x \rho U_x^\dagger$ for some `probe' state $\rho$.  The Holevo quantity then has the upper bound~\cite{HallPRX,Nair}
\begin{align} \label{asymmbound}
\chi(\E) &= H(\rho_\E) - H(\rho) 
\leq A^G(\rho) := H(\rho_G) - H(\rho) , ~~~~ \rho_G := \sum_n P_k\rho P_k =\sum_k P_k\rho_\E P_k, 
\end{align}
where the inequality follows because the decoherence map $\rho\rightarrow\rho_G$ is unital and hence entropy-increasing. The upper bound, $A^G(\rho)$, is called the asymmetry of $\rho$ with respect to $G$~\cite{vacc1}, and more generally extends to a resource measure for groups of unitary displacements of a given state~\cite{vacc1,vacc2, gour}.   Note that inequality~(\ref{asymmbound}) unifies and generalises  energy-time uncertainty relations~(8), (12) and~(E11) of Reference~\cite{Colestime}, which correspond to discrete prior distributions and  uniform continuous prior distributions (see also Section~\ref{sec:con}). 

For any estimate $A=X_{\rm est}$ of the shift parameter $X$ , one also has a simple lower bound for the Shannon mutual information from rate-distortion theory ~\cite{HallPRX,Nair}:
\beq \label{shanlower}
I(X_{\rm est}:X) = H(X) - H(X|X_{\rm est}) =H(X)-H(X-X_{\rm est}|X_{\rm est}) \geq H(X)-H(X-X_{\rm est}).
\eeq
Combining this with the Holevo bound~(\ref{hol}) and the asymmetry bound~(\ref{asymmbound}), and noting that $H(Z)=H(-Z)$ in general, then gives the strong metrological tradeoff relation~\cite{HallPRX}
\beq \label{metstrong}
 H(X_{\rm est}-X|\rho) + H(\rho_G) \geq H(X) + H(\rho) .
\eeq
This will be generalised to Renyi entropies further below.

For example, if $\rho$ is the state of an optical mode then phase shifts are generated by $G=N$, which is nondegenerate. Hence $H(\rho_N)=H(N|\rho)$, and the above relation becomes
\beq \label{metstrongphase}
H(\Theta_{\rm est}-\Theta|\rho) + H(N|\rho) \geq H(\Theta) + H(\rho) .
\eeq
Note that this implies and hence is stronger than both Equation~(\ref{urmet}) and Theorem~\ref{thm1} (for $\alpha=\beta=1$), which correspond to the special case of a uniformly random phase shift, with $p(\theta)=\frac{1}{2\pi}$ and $H(\Theta)=\log 2\pi$. Further, entropic number-phase uncertainty relation~(\ref{ur}) is recovered by choosing $\Theta_{\rm est}=\Phi$ and $p(\theta)=\frac{1}{2\pi}$ in Equation~(\ref{metstrongphase}), and hence may be viewed as a consequence of the Holevo bound for quantum communication channels~\cite{HallJPA}.

Equation~({\ref{metstrongphase}) also strengthens the Heisenberg limit in Equation~({\ref{rmse}) to take into account the degree of purity of the probe state and any prior information about the phase shift, via
\beq \label{rmsestrong}
{\rm RMSE}\geq (2\pi e)^{-1/2}\,e^{H(\rho)-H(N|\rho)}e^{H(\Theta)}  \geq \frac{(2\pi e^3)^{-1/2}\, e^{H(\rho)}e^{H(\Theta)}}{\langle N\rangle +\half}  
\eeq
(using units of nats for the entropies). The first lower bound follows via the well-known bound $H\leq \log {\rm RMSE} +\half\log 2\pi e$ for Shannon entropy, and the second from inequality~(\ref{upperlappb}) in Appendix~\ref{appb} (for $\alpha=\beta=1$). In particular, the lower bounds increase for less pure probe states (via increasing $H(\rho)$), and decrease for more prior information (via decreasing $H(\Theta)$). 

Finally, note that it is the latter property that underlies the different scaling behaviour of the Fisher bound~(\ref{fisher}), discussed in Section~\ref{sec:ent}. In particular, the first inequality in Equation~(\ref{rmsestrong}) implies that the only way to obtain a scaling of ${\rm RMSE}\sim 1/n_{\max}$, for the probe state $2^{-1/2}(|0\rangle+|n_{\max}\rangle)$, is if the entropy of the prior probability density $p(\theta)$ scales as $H(\Theta)\sim-\log n_{\max}$, i.e, if the phase shift is already known to within an accuracy of $L(\Theta)\sim 1/n_{\max}$ {\it before} any estimate is made. This is consistent with the Fisher bound, since the latter only applies  to the error of an (unbiased) estimate of a {\it known} phase shift~$\theta$. If the phase shift is not known, then ${\rm RMSE}_\theta$ must be generalised to take the prior distribution into account, e.g., via the ${\rm RMSE}$, for which the typically stronger bounds in Equation~(\ref{rmsestrong}) apply (to both biased and unbiased estimates). 

To emphasise this point of distinction, note that applying a phase shift $\theta$ to the probe state $2^{-1/2}(|0\rangle+|n_{\max}\rangle)$ gives the phase-shifted state $2^{-1/2}(|0\rangle+e^{-in_{\max}\theta}|n_{\max}\rangle)$, implying that no measurement on this state (nor on multiple copies thereof) can discriminate between the phase shifts $\theta, \theta+2\pi/n_{\max}, \dots, \theta+2\pi(n_{\max}-1)/n_{\max}$. Thus  $\theta$ cannot be accurately estimated to within an error of $1/n_{\max}$ via this probe state, as per the Fisher bound, unless it is already known to lie within an interval of length $2\pi/n_{\max}$. Finally, it is also worth noting that Fisher information bounds cannot be used to obtain Heisenberg limits in terms of average photon number, as per Equation~(\ref{rmsestrong}),unless further assumptions are made (in addition to local unbiasedness). For example, the probe state $\frac{\sqrt{3}}{2}\sum_m 2^{-m}|2^m\rangle$ has $\Delta N=\infty$ and $\langle N\rangle =3/2$~\cite{ZZ}, giving a trivial Fisher bound of 0 in Equation~({\ref{fisher}) but a nontrivial bound in Equation~(\ref{rmsestrong}).
Further discussion, including the case of multiple probe states, may be found in References~\cite{HallPRX,HallJapan,HallJPA}.

\subsection{Sandwiched R\'enyi relative entropy and mutual information}
\label{sec:dalpha}

The use of information bounds for quantum communication channels to obtain strong metrological bounds~(\ref{metstrong}) and~(\ref{metstrongphase}) points the way to strengthening the results in Section~\ref{sec:ent}. To proceed, however, suitable generalisations of mutual information and the Holevo quantity, i.e., of $I(A:X)$ and $\chi(\E)$, are required. These are provided via sandwiched R\'enyi relative entropies, as discussed below.

The starting point is  rewrite $I(A:X)$ and $\chi(\E)$ in terms of relative entropies:
\beq \label{holrewrite}
I(A:X) = D(p_{AX}\| p_A p_X) \leq D(\rho_{\E X}\| \rho_\E\otimes \rho_X) = \chi(\E) .
\eeq
Here $p_{AX}(a,x)=p(x)p(a|\rho_x)$ denotes the joint distribution for outcome $A=a$ and signal state $\rho_x)$, with marginals $p_A(a)=p(a|\rho_\E)$ and $p_X(x)=p(x)$; $\rho_{\E X}$ denotes the joint density operator $\int dx\,\rho_x\otimes |x\rangle\langle x|$, with reduced density operators $\rho_\E$ and $\rho_X= \int dx\, p(x) |x\rangle\langle x|$ for some orthogonal basis set $\{|x\rangle\}$ on an auxiliary Hilbert space; and
\beq \label{dpq}
D(p\| q):=\int dz\, p(z)[\log p(z)(\log p(z) - \log q(z)]\geq0,~~~ D(\rho\|\sigma):=\tr[\rho(\log \rho-\log\sigma]\geq0,
\eeq
are the respective Shannon and von Neumann relative entropies for two probability distributions $p,q$ and two density operators  $\rho,\sigma$ (defined to be infinite if the support of $p$ does not lie in the support of $q$ and similarly for $\rho$ and $\sigma$). The inequality in Equation~(\ref{holrewrite}), corresponding to the Holevo bound~(\ref{hol}), is a direct consequence of the data processing inequality~\cite{inftext,Ozawa}.

Now, the relative entropies appearing in Equation~(\ref{holrewrite}) have the alternate form
\beq
D(p_{AX}\| p_A p_X) =\inf_{q_A} D(p_{AX}\| q_A p_X),\qquad D(\rho_{\E X}\| \rho_\E\otimes \rho_X)= \inf_{\sigma_\E} D(\rho_{\E X}\| \sigma_\E\otimes\rho_X) 
\eeq
(using, e.g., $D(\rho_{\E X}\| \sigma_\E\otimes\rho_X) = D(\rho_{\E X}\| \rho_\E\otimes \rho_X)+D(\rho_\E\|\sigma_\E)$). Hence, if suitable definitions of classical and quantum R\'enyi relative entropies $D_\alpha(p\| q)$ and $D_\alpha(\rho\|\sigma)$ are available and satisfy a data processing inequality, then one can define a classical R\'enyi mutual information $I_\alpha(A:X)$~\cite{Sibson,Verdu} and associated Holevo quantity $\chi_\alpha(\E)$~\cite{Sharma,Wilde,Tom} (also called the quantum R\'enyi mutual information), satisfying a R\'enyi Holevo bound, via
\beq \label{ialpha}
I_\alpha(A:X) :=\inf_{q_A} D_\alpha(p_{AX}\| q_A p_X) \leq   \inf_{\sigma_\E} D_\alpha(\rho_{\E X}\| \sigma_\E\otimes\rho_X) =: \chi_\alpha(\E).
\eeq
Further, suitable definitions of $D_\alpha(p\| q)$ and $D_\alpha(\rho\|\sigma)$ do indeed exist,  given by the classical R\'enyi relative entropy~\cite{Renyi,Harremoes}
\beq \label{renyiclass}
D_\alpha(p\|q):=\frac{1}{\alpha-1} \log \int dz\, p(z)^\alpha q(z)^{1-\alpha} \geq 0
\eeq
and the quantum sandwiched R\'enyi relative entropy~\cite{Wilde,Lennert} 
\beq \label{sandwich}
D_\alpha(\rho\|\sigma) := \frac{1}{\alpha-1} \log \tr[(\sigma^{\frac{1-\alpha}{2\alpha}} \rho
\sigma^{\frac{1-\alpha}{2\alpha}})^\alpha ] \geq 0 .
\eeq
These reduce to standard relative entropies in the limit $\alpha\rightarrow1$, vanish for $p= q$ and for $\rho=\sigma$, and satisfy the data processing inequality 
\beq \label{data}
D_\alpha(\nu(\rho)\| \nu(\sigma)) \leq D_\alpha(\rho\|\sigma), \qquad \alpha\geq \half
\eeq
for all completely positive trace preserving maps ~$\nu$~\cite{Liebdata}. 
Further properties and operational interpretations of sandwiched relative entropies are given in References~\cite{Wilde, Tom, Lennert, Liebdata, Beigi,Mosonyi}. 

Finally, note that while the rewriting of $\chi(\E)$ in Equation~(\ref{holrewrite}) and the definition of $\chi_\alpha(\E)$ in Equation~(\ref{ialpha}) are perfectly valid for the case of a discrete set of signal states (with integration replaced by summation), there is an important point of rigour to be considered for the case of a continuous set of signal states. In particular, kets $\{|x\rangle\}$ forming an orthogonal basis set for this case are not normalisable, with $\langle x| x'\rangle=\delta(x-x')$, so that $\rho_X$ and $\rho_{\E X}$ are not well-defined density operators. This point may be addressed by working with discrete signal ensembles, with $p(x)$ and $\rho_x$ replaced by $p_j$ and $\rho_j$, via $p_j\rho_j= \int_{X_j} dx\,p(x)\rho_x$ for some countable (possibly finite) partition $\{X_j\}$ of the range of $x$.  The existence of a suitable orthonormal basis $\{|x_j\rangle\}$ is then  assured; the integrals defining  $\rho_X$ and $\rho_{\E X}$ can  be replaced by summations over the index $j$; and $\chi_\alpha(\E)$ in Equation~(\ref{ialpha}) rigorously defined as the limit (or supremum) of a suitable sequence of such partitions~\cite{endnote}. 
An approach of this type is used, for example, in Reference~\cite{Colestime}.   Hence, in what follows, the notation in the definition of $\chi_\alpha(\E)$ in Equation~(\ref{ialpha})  will be informally used for both discrete and continuous ensembles, under the implicit assumption that the limiting approach is  applied for the continuous case.

\subsection{R\'enyi asymmetry and upper bounds for mutual information}
\label{sec:asymm}

To generalise the asymmetry bound in Equation~(\ref{asymmbound}),   one may follow the general approach of Gao {\it et al.}~\cite{renyiasymm} to define the R\'enyi asymmetry of state $\rho$, with respect to any Hermitian operator $G$ with a discrete spectrum, by~\cite{Colestime},
\beq \label{renyiasymm}
A^G_\alpha(\rho) := \inf_{\sigma:[\sigma,G]=0} D_\alpha(\rho\|\sigma) ,\qquad \alpha\geq\half .
\eeq
This reduces to $A^G(\rho)$ in Equation~(\ref{asymmbound}) for $\alpha=1$, via the identity $D(\rho\|\sigma)= H(\rho_G)-H(\rho)+D(\rho_G\|\sigma)$ for $[\sigma,G]=0$. For later purposes, it is worth noting that, since any state commuting with $G$ must also commute with $h(G)$ for any function $h$, 
\beq \label{nonlinear}
A^{h(G)}_\alpha(\rho) =\inf_{\sigma:[\sigma,h(G)]=0} D_\alpha(\rho\|\sigma) \leq \inf_{\sigma:[\sigma,G]=0} D_\alpha(\rho\|\sigma) =A_\alpha^G(\rho).
\eeq

If relative entropies are regarded as quasi-distances, then the R\'enyi asymmetry may be interpreted as the distance from $\rho$ to the closest state that is invariant under the unitary transformation $U_x$ for all $x$, i.e, that commutes with $G$.  In particular, it vanishes if $\rho$ itself is invariant  i.e., if $[\rho,G]=0$, implying that the asymmetry may also be interpreted as the inherent `quantum' uncertainty of $G$~\cite{Kor} (see also Section~\ref{sec:con}). Further, when $G$ is nondegenerate, the asymmetry is a measure of coherence, relative to the eigenbasis of $G$~\cite{Chitambar}. This measure corresponds to the relative entropy of coherence for $\alpha=1$~\cite{coh}, and it is of interest to note that 
\beq \label{agcoh}
A_{1/2}^G(\rho) = -\log [1-C_g(\rho)],\qquad A_{\infty}^G(\rho) = \log [1+C_R(\rho)],
\eeq
for nondegenerate $G$, where $C_g(\rho)$ and $C_R(\rho)$ are the geometric coherence~\cite{coh} and robustness of coherence~\cite{coh,Chitambar}, respectively.
However, the utility of asymmetry for the purposes of this paper arises from the following upper bounds on the R\'enyi Holevo quantity.

\begin{Theorem} \label{thm3}
For any Hermitian operator $G$ having a discrete spectral decomposition $G=\sum_k g_k P_k$,   the R\'enyi asymmetry $A^G_\alpha(\rho)$ of an ensemble of signal states $\E\equiv\{\rho_x=e^{-ixG}\rho e^{ixG};p(x)\}$  has the upper and lower bounds
\beq
\label{thm3a}
\chi_\alpha(\E) \leq  A^G_\alpha(\rho) 
\leq H_\beta(G|\rho),\qquad \frac{1}{\alpha}+\frac{1}{\beta}=2,
\eeq
where $\chi_\alpha(\E)$ is the R\'enyi Holevo quantity in Equation~(\ref{ialpha})  and $H_\beta(G|\rho)$ is the classical R\'enyi entropy of the probability distribution $p(g_k|\rho)=\tr[\rho P_k]$. The upper bound is saturated when $\rho$ is pure.
\end{Theorem}

Theorem~\ref{thm3} clearly generalises the information bounds in Equation~(\ref{asymmbound}) for standard entropies, which correspond to the special case $\alpha=1$. The lower bound $\chi_\alpha(\E) \leq A^G_\alpha(\rho)$  also generalises energy-time uncertainty relation~(10) of Reference~\cite{Colestime}, from the special case of a uniform prior distribution to arbitrary prior distributions (see also Section~\ref{sec:con}). Similarly to Reference~\cite{Colestime}, the lower bound is a simple consequence of the fact that the transformation mapping $\rho_x\otimes|x\rangle\langle x|$ to $\rho\otimes|x\rangle\langle x|$ is a reversible isometry. In particular, data processing inequality~(\ref{data}) implies that $D_\alpha(\rho_{\E X} \|\sigma_\E\otimes\rho_X)=D_\alpha(U\rho_{\E X} U^\dagger\|U\sigma_\E\otimes\rho_X U^\dagger)$ for the unitary transformation $U=\int dx\,U_x^\dagger\otimes |x\rangle\langle x|$. This transformation maps $\rho_{\E X}\rightarrow \rho\otimes \rho_X$, and $\sigma_\E\otimes\rho_X\rightarrow\sigma_\E\otimes \rho_X$ when $[\sigma_\E,G]=0$, and hence
\beq
\chi_\alpha(\E)\leq  \inf_{\sigma_\E:[\sigma_\E,G]=0} D_\alpha(\rho_{\E X}\|\sigma_\E\otimes\rho_X) = \inf_{\sigma_\E:[\sigma_\E,G]=0} D_\alpha(\rho\otimes\rho_X\|\sigma_\E\otimes\rho_X) = A_\alpha^G(\rho) ,
\eeq
as desired, where the inequality follows immediately from the definition of $\chi_\alpha(\E)$ in Equation~(\ref{ialpha}).
The remainder of Theorem~\ref{thm3} is proved in Appendix~\ref{appc}. 

The upper bound for asymmetry in Equation~(\ref{thm3a}) has the benefit of being simpler to calculate than the asymmetry itself, and will be seen to underlie the constraint on $\alpha$ and $\beta$ in R\'enyi uncertainty relation~(\ref{urrenyi}). It is a consequence of the duality property
\beq \label{dual}
A_\alpha^G(|\psi\rangle\langle\psi|) = H_\beta(|\psi\rangle\langle\psi|_G), \qquad \frac{1}{\alpha}+\frac{1}{\beta}=2,
\eeq
for the asymmetry of pure states, also proved in Appendix~\ref{appc}, which in turn is formally connected to a deeper (and much harder to prove) duality relation for conditional R\'enyi entropies~\cite{Lennert,Beigi}.

Finally, it is of interest to note there is an alternative representation of asymmetry that relates it  more directly  to the Holevo quantity $\chi_\alpha$.  In particular, defining the continuous ensemble $\E_r\equiv \{\rho_x;p_r(x)\}$ by $p_r(x):=1/(2r)$ for $|x|<r$ and vanishing otherwise, it follows that $\lim_{r\rightarrow\infty}\rho_{\E_r}=\rho_G$ (see Appendix~\ref{appc}). Hence, using Equation~({\ref{hol}),
$\lim_{r\rightarrow\infty} \chi(\E_r) = S(\rho_G)-S(\rho) = A^G(\rho)$.
Thus the standard asymmetry is equal to the Holevo quantity in the limiting case of a maximally uniform ensemble of signal states~\cite{HallPRX}.   This property generalises to all R\'enyi asymmetries, i.e.,
\beq \label{infinity}
\chi_\alpha^\infty := \lim_{r\rightarrow\infty} \chi_\alpha(\E_r) = A_\alpha^G(\rho) ,
\eeq
as shown in Appendix~\ref{appc}. In the case that $U_x$ is periodic with period $x_p$, the limiting ensemble can be replaced by the uniform ensemble over $[0,x_p)$, as noted in References~\cite{vacc1,vacc2,gour} for $\alpha=1$.

\subsection{A convolution lower bound for mutual information}

Unlike the Holevo quantity $\chi_\alpha(\E)$, the classical R\'enyi mutual information $I_\alpha(A:X)$ in Equation~(\ref{ialpha})  can be calculated explicitly~\cite{Sibson,Verdu}. In particular, Equations~(\ref{ialpha}) and~(\ref{renyiclass}) give
\beq
I_\alpha(A:X) = \inf_{q_A}\frac{1}{\alpha-1}\log \int da\, q_A(a)^{1-\alpha}r_\alpha(a) ,\qquad 
r_\alpha(a):=\int dx\, p(x) p(a|x)^\alpha ,
\eeq
with integration replaced by summation over discrete ranges of $A$ and $X$, and  $p(a|x)=\tr[\rho_x M_a]$ for a measurement of observable $A$ described by the POVM $\{M_a\}$. A straightforward variation with respect to $q_A$ under the constraint $\int da\,q(a)=1$ then yields~\cite{Sibson,Verdu}
\beq
I_\alpha(A:X)= H_{\frac{1}{\alpha}}(\tilde p_A) + \frac{1}{1-\alpha}\log \int da\, r_\alpha(a) ,\qquad \tilde p_A(a):= r_\alpha(a){\Big/}\int da\,r_\alpha(a) .
\eeq
It may be checked that this expression reduces to $I(A:X)$ in Equation~(\ref{iax}) in the limit $\alpha\rightarrow1$, as expected.

However, while this expression allows the mutual information to be explicitly calculated for arbitrary observables $A$, the statistical characterisation of the error $X_{\rm est}-X$ of some estimate $X_{\rm est}$ of $X$ requires an expression involving the error probability density 
\beq \label{perrgen}
p_{\rm err}(y)=\int dx\, p_{X_{\rm est}X}(x+y,x) , 
\eeq
where $p_{X_{\rm est}X}(x_{\rm est},x)=p(x)\tr[\rho_x M_{x_{\rm est}}]$ is the joint probability density of $X_{\rm est}$ and $X$ (akin to Equation~(\ref{perr}) of Appendix~\ref{appa}). For Shannon mutual information this requirement is achieved via the lower bound in Equation~(\ref{shanlower}), which is partially generalised here via the following theorem and corollary. 

\begin{Theorem} \label{thm4}
The classical R\'enyi mutual information, for an estimate $X_{\rm est}$ of $X$ made on the ensemble $\E\equiv\{\rho_x;p(x)\}$, has the lower bound
\beq \label{thm4a}
I_\alpha(X_{\rm est}:X) \geq \inf_q D_\alpha(p_{\rm err}\|q\ast p_-), \qquad \alpha\geq \half,
\eeq
where $p_{\rm err}(y)$ is the probability density of the error $X_{\rm est}-X$;  $(f\ast g)(y)=\int dx\,f(x)g(y-x)$ denotes the convolution of functions $f$ and $g$; and $p_-(x):=p(-x)$. 
\end{Theorem}

Theorem~\ref{thm4} is proved in Appendix~\ref{appd}, and the lower bound in Equation~(\ref{thm4a}) will be referred to as the convolution lower bound. 
This bound has the desirable property of depending only on the prior probability density $p(x)$ and the error probability density $p_{\rm err}$, similarly to Equation~(\ref{shanlower}) for Shannon mutual information. 

\begin{Corollary} \label{cor2}
If the prior probability density $p(x)$ is uniform on an interval $I$ of length $\ell_I$, and vanishes outside this interval, then 
\beq \label{cor2a}
I_\alpha(X_{\rm est}:X) \geq \log \ell_I - H_\alpha(X_{\rm est}-X|\rho), \qquad \alpha\geq \half .
\eeq
\end{Corollary}

Note for $\alpha=1$ that this result is equivalent to Equation~(\ref{shanlower}) in the case that $p(x)$ is uniform on some interval $I$. Corollary~\ref{cor2} is  proved in Appendix~\ref{appd}, and relies on the following Lemma, which is of some interest in its own right.

\begin{Lemma} \label{lem}
If $Z$ is a random variable on the real line, and $\tilde Z= Z \mod I$ for some interval $I$, then 
\beq \label{lema}
H_\alpha(\tilde Z) \leq H_\alpha(Z) , \qquad \alpha\geq 0.
\eeq
\end{Lemma}

This lemma corresponds to the intuition that concentrating a probability density onto an interval will reduce the spread of distribution, and is proved in the last part of Appendix~\ref{appd}.   Note that if $Z$ is periodic, with period equal to $\ell_I$, then Lemma~\ref{lem} holds trivially with equality (noting that entropies are translation invariant).

\subsection{Putting it all together: strengthened  metrology bounds and uncertainty relations}
\label{sec:together}

Combining the R\'enyi Holevo bound, $I_\alpha(X_{\rm est}:X)\leq \chi_\alpha(\E)$ in Equation~(\ref{ialpha}), with the upper bounds for $\chi_\alpha(\E)$ in Theorem~\ref{thm3} and the lower bound for $I\alpha(X_{\rm est}:X)$ in Corollary~\ref{cor2} yields the inequality chain
\beq \label{funchain}
\log \ell_I - H_\alpha( X_{\rm est}-X|\rho)\leq I_\alpha(X_{\rm est}:X)\leq \chi_\alpha(\E) \leq A^G_\alpha(\rho) \leq H_\beta(G|\rho), \qquad \frac{1}{\alpha}+\frac{1}{\beta}=2 ,
\eeq
connecting the entropy of the estimation error to the asymmetry and entropy of the generator $G$, for the case of a prior distribution of the shift parameter that is uniform over some interval $I$. 
This chain immediately implies the following general result.

\begin{Theorem} \label{thm5}
For any estimate $X_{\rm est}$ of a unitary displacement $X$ applied to a probe state $\rho$, where the displacement is generated by a Hermitian operator $G$ with a discrete spectrum and has a uniform prior distribution over an interval $I$ of length $\ell_I$,  the entropy of the estimation error $X_{\rm est}-X$ satisfies the tradeoff relation
\beq \label{thm5a}
H_\alpha(X_{\rm est}-X|\rho) + A^G_\alpha(\rho) \geq \log \ell_I ,\qquad \alpha\geq\half .
\eeq
Further, for any nonlinear displacement generated by $h(G)$, the tradeoff relation
\beq \label{thm5b}
H_\alpha(X_{\rm est}-X|\rho) + H_\beta(G|\rho)\geq \log \ell_I, \qquad \frac{1}{\alpha}+\frac{1}{\beta}=2 ,
\eeq
holds for any function $h$.
\end{Theorem}

Equation~(\ref{thm5a}) of the theorem  generalises the metrology bounds and uncertainty relations in Section~\ref{sec:ent} to estimates of general unitary displacements, and strengthens them to take into account prior information and any nonpurity of the state.  Equation~(\ref{thm5b}), following via asymmetry property~(\ref{nonlinear}), shows that a nonlinear generator gives no advantage, in the sense that the R\'enyi entropy of the estimation error is lower-bounded by $\log\ell_I - H_\beta(G|\rho)$ for all generators $h(G)$, similarly to the case of Shannon entropy~\cite{HallPRX}. 

The content of Theorem~\ref{thm5} has the advantage of being independent of, and hence not requiring, any interpretation of the R\'enyi mutual information $I_\alpha$ and the Holevo quantity $\chi_\alpha$. The following four corollaries provide simple applications of the theorem to ${\rm RMSE}$ and to phase shifts, with further applications discussed in Section~\ref{sec:time}.

\begin{Corollary} \label{cor5}
	For any estimate $X_{\rm est}$ of a unitary displacement $X$ applied to a probe state $\rho$, where the displacement is generated by a Hermitian operator $G$ with  discrete spectral decomposition $\sum_k g_kP_k$, and $X$ has a uniform prior distribution over an interval $I$ of length $\ell_I$, the root-mean-square error of the estimate, RMSE=$\langle (X_{\rm est}-X)^2\rangle^{1/2}$, has the lower bounds
	\beq \label{cor5a}
	{\rm RMSE} \geq  \frac{\alpha^{\frac{\alpha}{\alpha-1}} f(\alpha) \ell_I}{2\pi}\, e^{-A^G_\alpha(\rho)} ,\qquad \alpha\geq \half,
	\eeq
	(using units of nats for $A^G_\alpha(\rho)$), where  $f(\alpha)$ is the function defined in Equation~(\ref{thm2b}). For the particular choices $\alpha=\infty$, $\alpha=\half$ and $\alpha\rightarrow1$, one further has
	\beq \label{cor5b}
	{\rm RMSE} \geq \frac{\ell_I}{2\sqrt{3}\,L_{1/2}(G|\rho)},\qquad {\rm RMSE}\geq \frac{\ell_I}{2\pi}\max_k p(g_k|\rho) \qquad {\rm RMSE} \geq \frac{\ell_I}{\sqrt{2\pi e}} e^{H(\rho) - H(\rho_G)} .
	\eeq
	where $L_\alpha(G|\rho)$ is a R\'enyi length as per Equation~(\ref{length}).
\end{Corollary}

The lower bound in Equation~(\ref{cor5a}) of Corollary~\ref{cor5} is sensitive to the purity of the probe state via $A_\alpha^G(\rho)$, and to the prior information via $\ell_I$, and follows by noting that the derivation of Equation~(\ref{funappb}) in Appendix~\ref{appb} goes through for the choice $\sigma^2= {\rm RMSE}$, and applying Theorem~\ref{thm5}. The first two inequalities in Equation~(\ref{cor5b}) then follow  via the upper bound for asymmetry in Equation~(\ref{funchain}), and generalise the corresponding bounds in Theorem~\ref{thm2} to include prior information about the estimate. The third inequality follows via the expression for asymmetry in Equation~(\ref{asymmbound}) for $\alpha=1$, and generalises Equation~(\ref{rmsestrong}) for phase and photon number to arbitrary discrete generators. Note that all lower bounds in Corollary~\ref{cor5} scale in proportion to the length of the interval, $\ell_I$, to which the displacement $X$ is restricted. Thus, better estimates are possible when more prior information is available.

\begin{Corollary} \label{cor3}
The root-mean-square error for any estimate of a phase shift $\Theta$ with a uniform prior distribution on an interval $I$ with length $\ell_I$, via a measurement on probe state $\rho$, satisfies the strong Heisenberg limit
\beq \label{cor3a}
{\rm RMSE}\geq  \frac{\ell_I}{2\pi} \, \frac{f_{\max}}{\langle N\rangle +\half} ,
\eeq
where $f_{\max}\approx 0.5823$ is the maximum value of the function $f(\alpha)$ in Equation~(\ref{thm2b}).
\end{Corollary}

Corollary~\ref{cor3} is a direct consequence of Equation~(\ref{cor5a})  in Corollary~\ref{cor5} for $G=N$, using $A_\alpha^G(\rho)\leq H_\beta(G|\rho)$ from Equation~(\ref{funchain}) and applying upper bound~(\ref{upperlappb}) from Appendix~\ref{appb}. It generalises the strong Heisenberg limit in Equations~(\ref{rmse}) and Theorem~\ref{thm2} to include prior information about the phase shift.

\begin{Corollary} \label{cor4}
The number and canonical phase observables $N$ and $\Phi$ satisfy the family of uncertainty relations previewed in Equation~(\ref{preview}), i.e.,
\beq \label{cor4a}
A_\alpha^N(\rho) + H_\alpha(\Phi|\rho) \geq \log 2\pi , \qquad \alpha\geq \half .
\eeq
\end{Corollary}

Corollary~\ref{cor4} follows from Theorem~\ref{thm5} for $G=N$, $\ell_I=2\pi$, and the particular choice of estimate $\Theta_{\rm est}=\Phi$, via $H(\Theta_{\rm est}-\Theta|\rho)=H(\Phi|\rho)$ on the unit circle.  For $\alpha=1$ it is equivalent to uncertainty relation~(\ref{ur}) for Shannon entropies when $\rho$ is a single mode state (noting that $A^G_1(\rho) = H(G|\rho)-H(\rho)$ for nondegenerate $G$ as per Equation~(\ref{asymmbound})), and more generally strengthens uncertainty relation~(\ref{urrenyi}) for R\'enyi entropies to take the degree of purity of the state into account.

\begin{Corollary} \label{cor6}
The number and canonical phase observables $N$ and $\Phi$  satisfy the family of uncertainty relations
\beq \label{cor6a}	A_\alpha^N(\rho) \Delta_\chi\Phi \geq \alpha^{\frac{\alpha}{\alpha-1}} f(\alpha), \qquad \alpha\geq\half,
	\eeq
where the function $f(\alpha)$ and the standard deviation $\Delta_\chi\Phi$ are defined in Equations~(\ref{thm2b}) and~(\ref{phasevar}).
\end{Corollary}

Corollary~\ref{cor6} follows via an analogous argument to the proof of Corollary~\ref{cor1} given in Appendix~\ref{appb}, but using the uncertainty relation of Corollary~\ref{cor4} in place of Equation~(\ref{urrenyi}).  It strengthens Equation~(\ref{cor1a}) of Corollary~\ref{cor1} when the state is non-pure.

\section{Applications to coherence measures, rotations, and energy-time tradeoffs}
\label{sec:time}

The results of the previous section have utility well beyond the number-phase examples given therein, as indicated here via several further applications.

\subsection{Coherence bounds}

Recall, as per the discussion in Section~\ref{sec:asymm}, that for nondegenerate $G$ the R\'enyi asymmetry $A_\alpha^G(\rho)$ is a measure of coherence, relative to the eigenbasis of $G$~\cite{Chitambar}.  Coherence measures are typically difficult to calculate explicitly, other than for pure states and general qubit states~\cite{Chitambar,coh}. However, the results of Section~\ref{sec:ent} lead to simple and strong upper and lower bounds for R\'enyi-related measures of coherence, as per the following theorem and corollary.

\begin{Theorem} \label{thm6} The R\'enyi entropy of coherence for state $\rho$, $C_\alpha(\rho):=A_\alpha^M(\rho)$, relative to a given orthonormal basis $\{|m\rangle\}$ indexed by a countable set of integers, has the upper and lower bounds
\beq \label{thm6a}	
\log 2\pi - H_\alpha(\Phi_{\boldsymbol\zeta}|\rho) \leq C_\alpha(\rho) \leq H_\beta(M|\rho), \qquad \frac{1}{\alpha}+\frac{1}{\beta}=2 ,
\eeq
where $M=\sum_m m|m\rangle\langle m|$, ${\boldsymbol\zeta}\equiv\{\zeta_m\}$ is an arbitrary set of reference phases, and $\Phi_{\boldsymbol\zeta}$ denotes the continuous phase observable corresponding to the POVM $\{|\phi\rangle\langle\phi|\}$ defined by
\beq \label{thm6b}
|\phi\rangle := \frac{1}{\sqrt{2\pi}} \sum_m e^{i \zeta_m} e^{-im\phi} |m\rangle .
\eeq
The upper bound is saturated for pure states.
\end{Theorem}

\begin{Corollary} \label{cor7}
The relative entropy of coherence, geometric coherence and robustness of coherence, with respect to basis $\{|m\rangle\}$, satisfy the respective bounds 
\begin{align} \label{cor7a}
 \log 2\pi - H(\Phi_{\boldsymbol\zeta}|\rho) &\leq C_{\rm rel~ent}(\rho) = H(M|\rho)-H(\rho),\\
 \label{cg}
  1-\frac{L_{1/2}(\Phi_{\boldsymbol\zeta}|\rho)}{2\pi} &\leq C_g(\rho) \leq 1 - \max_m \langle m|\rho|m\rangle, \\
  \label{cr}
2\pi\,\sup_{\phi\in[0,2\pi)} p(\phi|\rho) - 1 &\leq C_R(\rho) \leq \Big(\sum_n \langle m|\rho|m\rangle^{1/2}\Big)^2 -1
\end{align}
for state $\rho$, where $L_\alpha$ denotes the R\'enyi length in Equation~(\ref{length}). The upper bounds are saturated for pure states.
\end{Corollary}

The lower bounds in Theorem~\ref{thm6} and Corollary~\ref{cor7} hold for both finite and infinite Hilbert spaces, and follow in direct analogy to Corollary~\ref{cor4}, noting that $M$ is invariant under $|m\rangle\rightarrow e^{i\zeta_m}|m\rangle$ and using Equations~(\ref{phi}), (\ref{asymmbound}) and~(\ref{agcoh}). Note also that $ H_\alpha(\Phi_{\boldsymbol\zeta}|\rho)$ may be replaced by $\inf_{\boldsymbol\zeta} H_\alpha(\Phi_{\boldsymbol\zeta}|\rho)$, and that a weaker lower bound, $C_\alpha(\rho)\geq\log\alpha^{\frac{\alpha}{\alpha-1}}f(\alpha)-\log \Delta_\chi\Phi_{\boldsymbol\zeta}$, follows via Corollary~\ref{cor6}. The upper bounds follow directly from Theorem~\ref{thm3}, using  Equations~(\ref{asymmbound}) and~(\ref{agcoh}) and recalling that $H_\beta(\rho_G)=H_\beta(G|\rho)$ for nondegenerate $G$. It follows that a low phase uncertainty implies high coherence, which in turn requires a large uncertainy in $G$.

The bounds are relatively strong. The lower bounds are tight for all mixtures of number states, i.e., with zero coherence (noting that $p(\phi|\rho)=(2\pi)^{-1}$ for such states). Further, the upper bounds are saturated for all pure states, and are the strongest possible upper bounds that depend only on the distribution $p_m=\langle m|\rho|m\rangle$, being saturated for the pure state $\sum_m \sqrt{p_m}|m\rangle$ in particular.  Thus, for example, for an optical mode and the choice $M=N$, the coherent phase states  $|v\rangle=(1-|v|^2)^{1/2}\sum_n v^n|n\rangle$ (with $|v|<1$) not only have excellent phase resolution properties in general~\cite{HallJMO}, but also have the highest possible coherence for  given average photon number. Moreover, for $\alpha=1$ the bounds are equivalent to uncertainty relation~(\ref{ur}) for Shannon entropies, and more generally imply, and hence are stronger than, the uncertainty relation
\beq \label{genphaseur}
H_\alpha(M|\rho) + H_\beta(\Phi_{\boldsymbol\zeta}|\rho) \geq \log 2\pi, \qquad \frac{1}{\alpha}+\frac{1}{\beta}=2 
\eeq
for R\'enyi entropies, generalising Equation~(\ref{urrenyi}) to all dimensions.  

The upper bound  for geometric coherence in Equation~(\ref{cg}) recovers the upper bound in Theorem~1 of Reference~\cite{Zhang}. Further, the lower bound is typically stronger than the corresponding bound in Reference~\cite{Zhang}, as it depends on both the diagonal and off-diagonal elements of $\rho$. For example, for the maximally coherent  qubit state $|\psi\rangle=\frac{1}{\sqrt{2}}(|0\rangle+|1\rangle)$ one obtains $C_g(\rho) \geq 1-8/\pi^2\approx 0.189$ from Equation~(\ref{cg}) (with $\zeta_n\equiv0$),  whereas Theorem~1 of Reference~\cite{Zhang} gives a trivial lower bound of zero. Note that $C_g(\rho)=\half$ for this state, recalling that upper bound~(\ref{cg}) is saturated for pure states.

For the coherence of robustness, the upper bound $C_R(\rho)\leq \sum_{m,m'} |\langle m|\rho|m'\rangle|-1$ in Reference~\cite{Piani} is stronger than the upper bound in Equation~(\ref{cr}) for nonpure states (since the latter follows from the former via the Schwarz inequality). However, as noted above, Equation~(\ref{cr}) gives the strongest possible upper bound that depends only on the distribution of $M$. Further, the lower bound  in Equation~(\ref{cr}) is typically stronger than the corresponding bound in Reference~\cite{Piani}. For example, for the maximally coherent qubit state $|\psi\rangle$ in the above paragraph, both the lower bound in Equation~(\ref{cr}) and the lower bound in Theorem~5 of Reference~\cite{Piani} are saturated, with a value of unity. However, for the coherent phase state $|v\rangle$ mentioned above, the lower bound in Equation~(\ref{cr}) is again saturated, with a value $2|v|/(1-|v|)$ following from Equation~(37a) of Reference~\cite{HallJMO}, whereas Theorem~5 of Reference~\cite{Piani} gives a trivial lower bound of zero. 
 It would be of interest to compare these two lower bounds further, and to investigate the lower bound in Theorem~\ref{thm6} more generally.

\subsection{Rotations}

Two important applications of quantum metrology are the estimation of phase shifts, generated by the photon number operator, and the estimation of rotation angles, generated by angular momentum. In the latter case, for example, the strength of a magnetic field may be estimated via the rotation of an ensemble of atomic spins~\cite{Toth}.

The formal differences between phase and rotation estimation are small. For example, rotations of  GHZ states of $M$ spin qubits, $\frac{1}{\sqrt{2}}(\otimes^M |\uparrow\rangle+\otimes^M |\downarrow\rangle)$, are formally equivalent to phase shifts of single-mode states $\frac{1}{\sqrt{2}}(|0\rangle+|M\rangle)$ discussed following Theorems~\ref{thm1} and~\ref{thm2}, and to phase shifts of  two-mode NOON states $\frac{1}{\sqrt{2}}(|M,0\rangle+|0,M\rangle)$~\cite{Toth}.  

In fact the only significant formal difference between phase estimation and rotation estimation is that the eigenvalues of the photon number operator $N$ are nonnegative integers, whereas the eigenvalues of an angular momentum component $J_z$ range over all positive and negative integers. Thus, for example, the optical phase kets in Equation~(\ref{phi}) for the canonical phase observable $\Phi$ are replaced by the rotation kets
\beq
|\phi_z\rangle:=\frac{1}{\sqrt{2\pi}}\sum_{j=-\infty}^\infty e^{-ij\phi_z} |j\rangle 
\eeq
for the rotation angle observable $\Phi_z$ conjugate to $J_z$, where $\{|j\rangle\}$ denote the eigenstates of $J_z$. Note that  $\Phi_z$ corresponds to a Hermitian operator, with $\langle\phi_z|\phi_z'\rangle=\delta(\phi_z-\phi_z')$. 

Given that the general result in Theorem~\ref{thm5} holds for general discrete generators $G$, and noting that $J_z$ and $N$ only differ in their range of eigenvalues, 
it follows that all results for phase shifts not directly dependent on properties of the range of $N$ yield corresponding results for rotations, via the replacement of $N$ by $J_z$ and $\Phi$ by $\Phi_z$. Thus, for example, R\'enyi uncertainty relation (\ref{urrenyi}) holds for angular momentum and angle~\cite{Maassen,BB2006}, as does Theorem~\ref{thm1} and the metrology tradeoff relation
\beq
 H_\alpha(\Theta_{\rm z, est}-\Theta_z|\rho) + A^{J_z}_\alpha(\rho) \geq \log \ell_I, \qquad \alpha\geq \half
\eeq
corresponding to Theorem~\ref{thm5}. It follows that the estimation error, as characterised by its R\'enyi entropy, can only be small if the corresponding asymmetry of the state is large. In particular, a GHZ state of $M$ spin qubits has a relatively low asymmmetry, with $A^{J_z}_\alpha(\rho)=\log 2$ via duality property~(\ref{dual}) for pure states, and hence has a relatively poor angular resolution~\cite{HallPRX}. 

Further, since the relation between RMSE and entropy in Equation~(\ref{funappb}) of Appendix~\ref{appb} holds independently of the generator, all results for the RMSE of phase shifts and the standard deviation $\Delta_\chi\Phi$ of optical phase that do not depend on the eigenvalues of $N$ yield corresponding results for rotations. Thus, for example, the first two inequalities in Theorem~\ref{thm2} and the lower bounds in Corollaries~\ref{cor5} and~\ref{cor6} hold for angular momentum and angle, as does, e.g., the uncertainty relation
\beq
\Delta_\chi\Phi_z \geq \max_j p(j|\rho),
\eeq
corresponding to the second equality in Equation~(\ref{cor1b}) of Corollary~\ref{cor1}. This imples, for example, that a GHZ state of $M$ spin qubits has a relatively large standard deviation, with $\Delta_\chi\Phi_z\geq \half$.

Indeed, the only cases in which earlier results for number and phase do not immediately translate into results for angular momentum and angle are those involving the average photon number, such as the third uncertainty relation in Corollary~\ref{cor1} and the strong Heisenberg limits in Equations~(\ref{rmse}), (\ref{rmsealpha}), Theorem~\ref{thm2} and Corollary~\ref{cor3}. This is because such results rely on the upper bound for R\'enyi length in Equation~(\ref{upperlappb}) in Appendix~\ref{appb}, which assumes positive eigenvalues and so must be modified for the case of angular momentum.  A suitable modification is given by
\beq \label{upperl}
L_\beta(J_z|\rho) \leq \alpha^{\frac{\alpha}{\alpha-1}} \left(2\langle|J_z|\rangle+\half \langle 0|\rho|0\rangle\right) \leq \alpha^{\frac{\alpha}{\alpha-1}}\left( 2\langle|J_z|\rangle+\half\right), \qquad \frac{1}{\alpha}+\frac{1}{\beta}=2 ,
\eeq
as shown in the last part of Appendix~\ref{appb}.  This leads immediately to the following result  for angle estimation, corresponding to Corollary~\ref{cor3}.

\begin{Corollary} \label{cor8}
	The root-mean-square error for any estimate of a rotation $\Theta$ with a uniform prior distribution on an interval $I$ with length $\ell_I$, via a measurement on probe state $\rho$, satisfies the strong Heisenberg limits
	\beq \label{cor8a}
	{\rm RMSE}\geq  \frac{\ell_I}{2\pi} \, \frac{f_{\max}}{2\langle |J_z|\rangle +\half\langle0|\rho|0\rangle} \geq \frac{\ell_I}{2\pi} \, \frac{f_{\max}}{2\langle |J_z|\rangle +\half},
	\eeq
	where $f_{\max}\approx 0.5823$ is the maximum value of the function $f(\alpha)$ in Equation~(\ref{thm2b}).
\end{Corollary}

The lower bounds in Corollary~\ref{cor8} improve on the Heisenberg limit given in endnote~[29] of Reference~\cite{HallPRA}, in both the numerators and denominators, in addition to including prior information about the rotation via the factor $\ell_I/(2\pi)$.

\subsection{Energy and time}

The results of Section~\ref{sec:sand} also apply straightforwardly to the time evolution of quantum systems with discrete energy levels. For example, any estimate $T_{\rm est}$ of a time translation $T$ generated by a Hamiltonian with discrete spectral decomposition $E=\sum_k E_k P_k$, for uniform prior probability density $p(t)=1/\ell_I$ over an interval of length $\ell_I$, satisfies the tradeoff relation
\beq \label{timegen}
H(T_{\rm est}-T|\rho) + A^E_\alpha(\rho) \geq \log \ell_I,\qquad \alpha\geq\half,
\eeq
as an immediate consequence of Theorem~\ref{thm5}. Further, the RMSE of the estimate satisfies the lower bounds in Corollary~\ref{cor5} for $G=E$.  

If the system is periodic, with period $\tau=2\pi/\omega$, the energy eigenvalues are of the form
\beq \label{eper}
E_k = \epsilon +  \hbar \omega n_k,
\eeq
where $\epsilon$ denotes the groundstate energy and the $n_k$ are nonnegative integers. Hence a time translation of the system by an amount $t$ is formally identical to a phase shift of an optical mode by $\phi=\omega t$ (for a state of the mode with support restricted to the number states $\{|n_k\rangle\}$), and all previous results for number and phase carry over immediately to analogous results for the energy and time of periodic systems via the replacement of $\phi$ by $\omega t$ and $N$ by $(E-\epsilon)/(\hbar\omega)$. For example, Corollary~\ref{cor3} implies the strong Heisenberg limit
\beq \label{heistime}
{\rm RMSE} = \langle (T_{\rm est}-T)^2\rangle^{1/2} \geq  \frac{\hbar\ell_I}{\tau} \, \frac{f_{\max}}{\langle E-\epsilon\rangle +\half\hbar\omega}
\eeq
for the RMSE for any estimate of the time shift of a periodic system, if the prior probability density $p(t)$ is uniform over an interval of length $\ell_I$, which strengthens the result in Reference~\cite{HallNJP} for this case. Similarly,  Corollaries~\ref{cor4} and~\ref{cor6} imply the strong energy-time uncertainty relations
\beq \label{uncertper}
A_\alpha^E(\rho) + H_\alpha({\cal T}_\tau|\rho)\geq \log\tau,\qquad A_\alpha^E(\rho) \Delta_{t_0}{\cal T}_\tau\geq \frac{\alpha^{\frac{\alpha}{\alpha-1}} f(\alpha)}{2\pi} \,\tau, \qquad \alpha\geq\half
\eeq
for the energy and canonical time observables of a periodic system with period $\tau$, where ${\cal T}_\tau$ denotes the canonical time observable  corresponding to $\Phi/\omega$~\cite{Holevo, Halltime}, and $\Delta_{t_0}{\cal T}_\tau=\langle ({\cal T}_\tau-t_0)^2\rangle^{1/2}$ is the standard deviation about any reference time $t_0$.  Note that the first of these relations, combined with the asymmetry bound in Theorem~\ref{thm3}, implies and hence is stronger than the known R\'enyi entropic uncertainty relation
\beq \label{uncertperweak}
H_\alpha(E|\rho) + H_\beta({\cal T}_\tau|\rho) \geq \log \tau, \qquad \frac{1}{\alpha}+\frac{1}{\beta}=2
\eeq
for the energy and time observables of periodic systems~\cite{Rastegin}, analogous to Equation~(\ref{urrenyi}) for number and phase.

While some quantum systems, such as harmonic oscillators and qubits, are indeed periodic,  most systems with discrete Hamiltonians do not have energy eigenvalues of the form in Equation~(\ref{eper}) and so are nonperiodic. This is not an issue for the basic energy-time metrology tradeoff relation~(\ref{timegen}), which is universal for discrete Hamiltonians and so applies equally well to both periodic and nonperiodic systems, as do the bounds for the RMSE of time estimates in Corollary~\ref{cor3} (choosing $G=E$). However, the question of whether there are  time-energy uncertainty relations for nonperiodic systems, that generalise Equations~(\ref{uncertper}) and~(\ref{uncertperweak}) for periodic systems, is less straightforward.

This question has been addressed for the case of Shannon entropies, via the definition of a canonical time observable $\cal T$ that is applicable to both periodic and nonperiodic systems. This observable has almost-periodic probability density $p_{ap}(t)$ associated with it,  and a corresponding almost-periodic Shannon entropy $H^{ap}({\cal T}|\rho)$, which satisfies the energy-time-energy entropic uncertainty relation~\cite{Halltime,HallJPA}
\beq \label{peruncert}
H(E|\rho) + H^{ap}({\cal T}|\rho) \geq 0 .
\eeq
This relation reduces to Equation~(\ref{uncertperweak}) for the case of periodic systems, and is strengthened  and extended below to general R\'enyi entropies.

To proceed, it is convenient to first deal with any energy degeneracies, by taking the degree of degeneracy to be the same for each energy eigenvalue (by formally extending the Hilbert space if necessary), so that the energy eigenstates can be formally written as $|E_k\rangle\otimes |d\rangle$ with the range of the degeneracy index $d$  independent of $E_k$. Defining $1_D:=\sum_d|d\rangle\langle d|$, the canonical time observable conjugate to $E=\sum_k |E_k\rangle\langle E_k|\otimes 1_D$ is then defined via the almost-periodic POVM ${\cal T}\equiv\{M_t\}$ given by
\beq \label{canont}
M_t := \sum_{k,k'}  e^{-i(E_k-E_{k'})t/\hbar}|E_k\rangle\langle E_{k'}| \otimes 1_D \geq 0, \qquad t\in(-\infty,\infty),
\eeq
and associated almost-periodic probability density
\beq \label{pap}
p_{ap}(t|\rho) := \tr[\rho M_t],\qquad t\in(-\infty,\infty).
\eeq
for state $\rho$~\cite{Halltime}. Note the formal similarity to the canonical phase observable in Section~\ref{sec:length}. For periodic systems, the associated periodic time observable ${\cal T}_\tau$ above has the related POVM $\{\tau^{-1}{\cal T}_t: t\in [0,\tau)\}$, with associated periodic probability density~\cite{Halltime} 
\beq \label{apper}
p_\tau(t|\rho) = \tau^{-1} p_{ap}(t|\rho)  ,\qquad t\in[0,\tau) .
\eeq

Now, it is easy to check that $p_{ap}(t|\rho)$ is not normalised with respect to the Lebesgue measure, and indeed that $\int_{-\infty}^\infty dt\,p_{ap}(t|\rho)$ diverges. Hence an alternative measure is required. This is provided by the Besicovitch measure $\mu_{ap}[\cdot]$, defined on the algebra of  almost-periodic functions,  i.e,  functions of the form $f(t)=\sum_j f_j e^{i\omega_j t}$ with $\sum_j |f_j|^2<\infty$, by~\cite{Bes}
\beq \label{limsup}
\mu_{ap}[f] := \limsup_{s\rightarrow\infty} \frac{1}{s} \int_0^s dt\, f(t)
\eeq
For $f(t)=p_{ap}(t|\rho)$ in Equation~(\ref{pap}) this yields
\beq
\mu_{ap}[p_{ap}] = \sum_{k,d}  \langle E_{k},d|\rho|E_k,d\rangle =1,
\eeq
and hence the almost-periodic density is normalised, as desired. The average of the  almost-periodic function $f(t)$ with respect to $p_{ap}(t|\rho)$can then be defined as~\cite{Halltime}
\beq
\langle f\rangle:=\mu_{ap}[p_{ap}f] =  \sum_{j,k,k': E_k-E_{k'}=\hbar\omega_j} f_j\, \langle E_{k'}|\tr_D[\rho]|E_k\rangle ,
\eeq
and the almost-periodic R\'enyi entropy of the canonical time observable by
\beq \label{hap}
H^{ap}_\alpha({\cal T}|\rho) := \frac{1}{1-\alpha} \log\mu_{ap}[(p_{ap})^\alpha] =  \frac{1}{1-\alpha} \log \limsup_{s\rightarrow\infty} \frac{1}{s} \int_0^s dt\,p_{ap}(t|\rho)^\alpha ,
\eeq
generalising the case of almost-periodic Shannon entropy~\cite{Halltime,HallJPA}.  For the case of a periodic system with period $\tau$, it  follows from Equation~(\ref{apper}) that
\beq \label{entdiff}
H^{ap}_\alpha({\cal T}|\rho) = H_\alpha({\cal T}_\tau|\rho) -\log \tau ,
\eeq
where $H_\alpha$ denotes the R\'enyi entropy of the periodic probability density $p_\tau(t|\rho)$.

Finally, any almost-periodic function $f(t)$ can be approximated by a periodic function to any desired accuracy, via a sequence of periodic functions $f_m(t)$ with  respective periods $\tau_m$, such that $\tau_m\rightarrow\infty$ and $f_m(t)$ converges uniformly  to $f(t)$ in the limit $m\rightarrow\infty$~\cite{Bes}. In particular, $p_{ap}(t|\rho)$ has such a sequence of periodic approximations, each corresponding to the canonical time distribution of a periodic system with energy observable $E^{(m)}=\sum_k E^{(m)}_k|E_k\rangle\langle E_k|\otimes 1_D$ and period $\tau_m$, with $E^{(m)}_k\rightarrow E_k$ as $m\rightarrow\infty$~\cite{HallPegg}. Further, each such periodic system must satisfy uncertainty relation~(\ref{uncertper}). Hence, using Equation~(\ref{entdiff}) and taking the limit $m\rightarrow\infty$, one obtains the following general result.

\begin{Corollary} \label{cor9}
The energy and almost-periodic canonical time observables $E$ and ${\cal T}$ satisfy the family of uncertainty relations
\beq \label{cor9a}
A^E_\alpha(\rho) + H^{ap}_\alpha({\cal T}|\rho) \geq 0 , \qquad \alpha\geq\half,
\eeq
for any quantum system with a discrete energy spectrum, where $H^{ap}_\alpha({\cal T}|\rho)$ is the almost-periodic R\'enyi entropy in Equation~(\ref{hap}).
\end{Corollary}

Corollary~\ref{cor9} generalises Corollary~\ref{cor4} and uncertainty relations~(\ref{uncertper}) and~(\ref{uncertperweak}), for periodic systems, to any system with a discrete energy spectrum. Moreover, using Equation~(\ref{asymmbound}) for the case $\alpha=1$~\cite{HallJPA}, and Theorem~\ref{thm3} more generally, the corollary further leads to the respective energy-time uncertainty relations
\beq
H(\rho_E) +  H^{ap}({\cal T}|\rho) \geq H(\rho), \qquad H_\alpha(E|\rho) + H^{ap}_\beta({\cal T}|\rho) \geq 0, \qquad \frac{1}{\alpha}+\frac{1}{\beta}=2 ,
\eeq
for the Shannon and R\'enyi entropies of general systems. It may be noted, however, that the function $f(t)=t^2$ is not  almost-periodic, implying that there is no analogue of $\Delta_{t_0}{\cal T}$ for  nonperiodic systems and hence no corresponding generalisation of the second uncertainty relation in Equation~(\ref{uncertper}).

Finally, it should be noted that a suggestion in Reference~\cite{Halltime}, to interpret $I^{ap}(\rho):=-H^{ap}({\cal T}|\rho)$ as the maximum information that can be gained  about a random time shift in the limit of a uniform prior distribution on the real line,  via a measurement of ${\cal T}$, is incorrect. This quantity is in fact a {\it lower} bound for the information gain in this scenario. In particular, note from Corollary~\ref{cor2} that
$I_\alpha({\cal T}_\tau:T) \geq \log \tau - H_\alpha({\cal T}_\tau-T|\rho)= \log \tau - H_\alpha({\cal T}_\tau|\rho)$
for uniformly random time displacements of a periodic system with period $\tau$. Hence, choosing the same sequence of periodic systems as above and using Equation~(\ref{entdiff}), the lower bound
\beq
I_\alpha({\cal T}:T)\geq -H^{ap}_\alpha({\cal T}|\rho) = I^{ap}_\alpha(\rho)
\eeq
follows for information gain in the limit of a uniform prior distribution on the real line, valid for both periodic and nonperiodic systems.  This lower bound can be quite strong for systems with many pairs of resonant energy levels (i.e., with $E_j-E_{j'}=E_k-E_{k'}\neq0$), but is no greater than $\log 2$ in the case of no shared resonances and $\alpha=1$~\cite{Halltime}.

\section{Discussion}
\label{sec:con}

The main results of the paper, embodied in Theorems~\ref{thm1}--\ref{thm6} and Corollaries~\ref{cor1}--\ref{cor9}, are seen to have wide applicability, including lower bounds for  the error of any estimate of unitary displacement parameters, such as phase shifts, rotations and time;  Heisenberg limits for the scaling of RMSE with average photon number and angular momentum;  upper and lower bounds for measures of coherence; and uncertainty relations for canonically conjugate observables. As demonstrated by various examples, the results are typically stronger than existing results in the literature.

Whereas the results in Section~\ref{sec:ent} are based on known uncertainty relation~(\ref{urrenyi}) for R\'enyi entropies, the results in Sections~\ref{sec:together} and~\ref{sec:time} rely on the upper and lower bounds in Theorems~\ref{thm3} and~\ref{thm4} for the R\'enyi mutual information of quantum communication channels, which provide a path to far stronger  metrology bounds and uncertainty relations. All of these results have the advantage of being independent of, and hence not requiring any interpretation of, the R\'enyi mutual information itself. Indeed a number of the inequalities in Theorem~\ref{thm2} and Corollaries~\ref{cor1}, \ref{cor3} and~\ref{cor8} do not refer even to R\'enyi entropies.

There is an interesting subtlety worth noting in regard to  entangled states. In particular,  if a unitary displacement acts only on one component of an entangled probe state, then there are two distinct scenarios: (i)~an estimate is made via a measurement on that component, or (ii)~via a measurement on the whole state. The first scenario, by limiting the class of measurements, will in general have an increased estimation error, that is not taken into account in Theorem~\ref{thm5} and its corollaries. Fortunately this is straightforward to remedy, by replacing the state $\rho$ in those results by its accessible component,  i.e, by the partial trace $\tr_R[\rho]$ over any unmeasured components. Stronger lower bounds for estimation error are thereby obtained in the first scenario, noting that $A^G_\alpha(\tr_R[\rho])\leq A_\alpha^G(\rho)$ via the data processing inequality, which yield correspondingly improved uncertainty relations for  observables  that act on a component of an entangled state. 

The R\'enyi asymmetry, already known to be a useful resource in various contexts~\cite{vacc1,vacc2,gour,Chitambar,Colestime} is seen to also be a valuable resource in quantum metrology. In particular, the various lower bounds for the estimation error of unitary displacements, whether measured via its entropy or the RMSE, decrease as the asymmetry increases, making probe states with high asymmetry desirable.  Moreover, the strong uncertainty relations derived in the paper, e.g., Corollaries~\ref{cor4}, \ref{cor6} and~\ref{cor9}, imply that $A^G_\alpha(\rho)$ can further be regarded as a measure of the intinsically `quantum' uncertainty of $G$ for state $\rho$, given that it only vanishes for eigenstates of $G$ and for any classical mixtures thereof. Together with Theorem~\ref{thm3}, this suggests that the `total uncertainty' of $G$ for state $\rho$, as characterised by ${\cal U}_\alpha^{\rm total}(G|\rho):=H_\alpha(G|\rho)$, can be decomposed into quantum and classical contributions via
\beq	
{\cal U}_\alpha^{\rm total}(G|\rho) = {\cal U}_\alpha^{\rm quantum}(G|\rho) +{\cal U}_\alpha^{\rm classical}(G|\rho) , \qquad \alpha\geq \half,
\eeq
where ${\cal U}_\alpha^{\rm quantum}(G|\rho):=A^G_\beta(\rho)$, ${\cal U}_\alpha^{\rm classical}(G|\rho):=H_\alpha(G|\rho)-A^G_\beta(\rho)$, and $1/\alpha+1/\beta=2$. It follows that the quantum contribution vanishes if and only if the state is classical with respect to $G$, i.e.,  $[G,\rho]=0$. Conversely, the classical contribution vanishes if and only if the state is pure, i.e., has no classical mixing. For $\alpha=1$ this decomposition matches the one introduced in Reference~\cite{Kor} for Shannon entropy. An analogous decomposition of variance into quantum and classical contributions has been given by Luo~\cite{Luo2005}.

It was mentioned in Section~\ref{sec:sand} that the universal asymmetry bound $\chi_\alpha(\E)\leq A_\alpha^G(\rho)$  in Equations~(\ref{asymmbound}) and~(\ref{thm3a}) unifies and generalises the recent energy-time estimation relations given by Coles {\it et al.}~\cite{Colestime}. To see this in more detail, note that the main result in Reference~\cite{Colestime} for uniform discrete ensembles, in the form given by  Equations~(3) and~(10) thereof, translates in the notation of this paper to
\beq
-\inf_{\sigma_\E} D_\alpha(\rho_{\E X}\|\sigma_\E\otimes 1_X) + A^G_\alpha(\rho) \geq \log d,
\eeq
where $\E$ is an ensemble corresponding to $d$ displaced states $\rho_{x_k}$ having uniform prior probabilities $p(x_k)=1/d$, and $1_X$ is the unit operator on a corresponding reference system for these displacements (with orthonormal basis $\{|x_k\rangle\}$ as in Section~\ref{sec:dalpha}). It is straightforward to check that $D_\alpha(\rho_{\E X}\|\sigma_\E\otimes 1_X)=D_\alpha(\rho_{\E X}\|\sigma_\E\otimes d^{-1}1_X) -\log d = D_\alpha(\rho_{\E X}\|\sigma_\E\otimes \rho_X) -\log d$ for such ensembles, implying via Equation~(\ref{ialpha}) that the above result  is equivalent to $\chi_\alpha(\E)\leq A^G_\alpha(\rho)$, as claimed.  A similar equivalence holds between the asymmetry bound and Equation~(12) of Reference~\cite{Colestime}  for uniform continuous ensembles.  Hence, while the estimation relations in Reference~\cite{Colestime} are interpreted via a game, with scores determined by R\'enyi conditional entropies for estimates of $G$ and of the displacements it generates, they may  also be interpreted as special cases of the asymmetry bound. Alternatively, they may be interpreted  via Equation~(\ref{infinity}) as instances of a general inequality  between the quantum R\'enyi mutual information of a given ensemble and that of the ensemble corresponding to the limit of a maximally uniform prior distribution.

Finally, several topics for future work are suggested by the results. These include using Theorem~\ref{thm4} to obtain explicit lower bounds for the mutual information of ensembles with arbitrary prior probabilities, including discrete prior probability distributions (thus generalising the results in Sections~\ref{sec:together} and~\ref{sec:time}); extending the analysis to displacements induced by generators with continuous spectra, such as translations generated by momentum, and to multiparameter displacements (some preliminary results for the case of Shannon entropy are given in Reference~\cite{HallJPA}); obtaining Heisenberg-type limits for RMSE in terms of the variance of $N$ rather than of $\langle N\rangle$ (via corresponding upper bounds on R\'enyi entropies analogous to Equation~(\ref{upperlappb}) of Appendix~\ref{appb}); and further investigating the lower bounds for coherence in Theorem~\ref{thm6} and~Corollary~\ref{cor7}.

\appendix

\section{}
\label{appa}

\begin{proof}[Proof of Theorem~\ref{thm1}]
	First, applying a uniformly random phase shift $\theta$ to a probe state $\rho$ gives the phase-shifted state $\rho_\theta=e^{-iN\theta}\rho e^{iN\theta}$, with associated prior probability density $p(\theta)=1/(2\pi)$. Further, any estimate $\theta_{\rm est}$ made of $\theta$ must be described by some POVM $\Theta_{\rm est}\equiv \{M_{\theta_{\rm est}}\}$, with $M_{\theta_{\rm est}}\geq0$ and $\oint d\theta_{\rm est}\, M_{\theta_{\rm est}}=\hat 1$. Hence the joint probability density of $\theta$ and $\theta_{\rm est}$ is given by
	\beq
	p(\theta,\theta_{\rm est}|\rho) = p(\theta) p(\theta_{\rm est}|\rho_\theta)  = \frac{1}{2\pi}\tr[\rho_\theta M_{\theta_{\rm est}}] ,
	\eeq
	and the error of the estimate, $\Theta_{\rm est}-\Theta$, has the corresponding marginal probability density 
	\begin{align}
		p(\theta_{\rm err}) &=  \oint d\theta_{\rm est}\, p(\theta_{\rm est}-\theta_{\rm err},\theta_{\rm est}|\rho) 
		=\frac{1}{2\pi}\oint d\theta_{\rm est}\, \tr[\rho_{\theta_{\rm est}-\theta_{\rm err}} M_{\theta_{\rm est}}] 
		= \tr[\rho \widetilde M_{\theta_{\rm err}}] , 
		\label{perr}
	\end{align}
	where
	\beq
	\widetilde M_{\theta_{\rm err}}:= e^{-iN\theta_{\rm err}}\widetilde M_0 e^{iN\theta_{\rm err}},\qquad \widetilde M_0:= \frac{1}{2\pi}\oint d\theta_{\rm est}\, e^{iN\theta_{\rm est}} M_{\theta_{\rm est}} e^{-iN\theta_{\rm est}} ,
	\eeq
	and the cyclic property of the trace has been used. Note the useful property $\langle n|\widetilde M_0|n\rangle=1$.
	
	Now, $\widetilde M_0\geq 0$ follows from $M_{\theta_{\rm est}}\geq0$, implying it can be written in the form $\widetilde M_0=\sum_{\tilde m} |\tilde m\rangle\langle \tilde m|$ for some set of kets $\{|\tilde m\rangle\}$. Defining the completely-positive trace-preserving map $\mu(\rho)=\sum_{\tilde m} A_{\tilde m}\rho A_{\tilde m}^\dagger$ with $A_{\tilde m}=\sum_{n=0}^\infty \langle \tilde m|n\rangle\,|n\rangle\langle n|$~\cite{HallJMO,Halltime}, it is then straightforward to calculate the canonical phase distribution of $\Phi$ for state $\mu(\rho)$ as
	\beq
	p(\phi|\mu(\rho)) = p(\theta_{\rm est} - \theta=\phi|\rho),
	\eeq
	i.e, it is identical to the distribution of the error $\Theta_{\rm est}-\Theta$. One further finds
	\beq
	p(n|\mu(\rho)) = \langle n|\mu(\rho)|n\rangle = \langle n|\rho|n\rangle = p(n|\rho) ,
	\eeq
	i.e., the number distributions of $\mu(\rho)$ and $\rho$ are identical. Finally, it may be checked that $\mu(\hat 1)=\hat 1$, i.e., $\mu$ is a unital map, implying that the von Neumann entropy increases under $\mu$~\cite{inftext}, i.e.,
	\beq \label{unital}
	H(\mu(\rho)) \geq H(\rho).
	\eeq
	Hence, uncertainty relation~(\ref{ur}) for standard entropies immediately implies that
	\beq
	H(\Theta_{\rm est}-\Theta|\rho) + H(N|\rho) = H(\Phi|\mu(\rho))+H(N|\mu(\rho)) \geq \log 2\pi+H(\mu(\rho)) \geq \log 2\pi + H(\rho)
	\eeq
	as per Equation~(\ref{urmet}), while the R\'enyi uncertainty relation~(\ref{urrenyi}) similarly yields 
	\beq
	H_\alpha(\Theta_{\rm est}-\Theta|\rho) + H_\beta(N|\rho) \geq \log 2\pi,\qquad 
	\frac{1}{\alpha}+\frac{1}{\beta}=2 ,
	\eeq
	as per Equation~(\ref{thm1b}) of Theorem~\ref{thm1}.
\end{proof}

\section{}
\label{appb}

\begin{proof}[Proof of Theorem~\ref{thm2}]
The proof proceeds by establishing upper bounds on R\'enyi entropies under various constraints. First, consider the variation of the quantity $\int dx \,p(x)^\alpha$ with respect to probability density $p(x)$ for $\alpha\geq 1/2$ and $\alpha\neq1$, under the constraints $\int dx\,x^2p(x) =\sigma^2$ and $\int dx\,p(x)=1$, where integration is over the real line. Applying the method of Lagrange multipliers then gives an extremal probability distribution of the form
\beq \label{pextappb}
p(x)=\frac{1}{\lambda} p_1(x/\lambda), \qquad p_1(x) = K(1\pm x^2)^{\frac{-1}{1-\alpha}} .
\eeq
where the $+$ sign is chosen for $1/2\leq\alpha<1$ and the $-$ sign for $\alpha>1$~\cite{maxrenyi}, and in the latter case one takes $p_1(x)=0$ for $|x|>1$. These choices correspond to a maximum and minimum of $\int dx \,p(x)^\alpha$, respectively, and hence to a maximum of the associated R\'enyi entropy. The values of $\lambda$ and $K$ are determined by the constraints as~\cite{maxrenyi}
\beq
\lambda= \sigma \left(\frac{3\alpha-1}{|1-\alpha|}\right)^{1/2}, \qquad
K=\left\{ \begin{array}{ll}
\frac{\Gamma(\frac{1}{1-\alpha})}{\sqrt{\pi}\,\Gamma(\frac{1}{1-\alpha}-\frac{1}{2})}, & 1/2\leq \alpha<1\\
\frac{\Gamma(\frac{\alpha}{\alpha-1}+\frac{1}{2})}{\sqrt{\pi}\,\Gamma(\frac{\alpha}{\alpha-1})}, & \alpha>1
\end{array} \right. .
\eeq
Since phase error is restricted to a subset of the real line, it follows that its R\'enyi entropy can be no greater than the maximum for $p(x)$ with $\sigma={\rm RMSE}$.  Hence,
\begin{align}
H_\alpha(\Theta_{\rm est}-\Theta|\rho) &\leq H_\alpha[p]=\log \lambda+ H_\alpha[p_1]=\log \lambda +\frac{1}{1-\alpha}\log \int dx\, p_1(x)^\alpha \nn\\
&= \log {\rm RMSE}+ \log \left[\left(\frac{3\alpha-1}{|1-\alpha|}\right)^{1/2} \left(\frac{2\alpha}{3\alpha-1}\right)^{\frac{1}{1-\alpha}} \right] -\log K,
\label{funappb}
\end{align}
where the last line follows by direct calculation. It may be checked that this inequality also holds in the limit $\alpha\rightarrow1$, for which it becomes equivalent to the standard bound $H\leq \log {\rm RMSE} +\half\log 2\pi e$ for Shannon entropy. Inverting the above result, and using $H_\alpha(\Theta_{\rm est}-\Theta|\rho)\geq \log 2\pi - H_\beta(N|\rho)$ from Equation~(\ref{thm1b}) of Theorem~\ref{thm1}, then gives the lower bound
\beq \label{rmseappb}
{\rm RMSE} \geq \frac{\alpha^{\frac{\alpha}{\alpha-1}} f(\alpha)}{L_\beta(N|\rho)},\qquad \frac{1}{\alpha}+\frac{1}{\beta} = 2 ,
	\eeq
where $f(\alpha)$ is defined in Equation~(\ref{thm2b}). 

To obtain the first lower bound in Theorem~\ref{thm2}, consider the limit $\alpha\rightarrow\infty$ in Equation~(\ref{rmseappb}) above. The constraint on $\alpha$ and $\beta$ gives $\beta=\frac{1}{2}$, and one finds
\beq
\lim_{\alpha\rightarrow\infty} \alpha^{\frac{\alpha}{\alpha-1}} f(\alpha) = 2\sqrt{\pi}\,\frac{\Gamma(\frac{3}{2})}{\Gamma(1)} \lim_{\alpha\rightarrow\infty} \alpha^{\frac{1}{\alpha-1}} \left(\frac{\alpha-1}{3\alpha -1}\right)^{1/2} = \frac{\pi}{\sqrt{3}}
\eeq
(corresponding to the maximum value of $\alpha^{\frac{\alpha}{\alpha-1}} f(\alpha)$), yielding the required first bound. Similarly, substitution of $\alpha=\half$ into Equation~(\ref{rmseappb}) gives unity for the numerator and $L_\infty(N|\rho)=(\max_n p(n|\rho))^{-1}$ for the denominator, yielding the second lower bound in Theorem~\ref{thm2}.

The third and final bound requires an upper bound for the R\'enyi length $L_\beta(N|\rho)$ as a function of $\langle N\rangle$. This is obtained by considering the variation of the quantity $L_\beta(X)=[\int dx\,p(x)^\beta]^{1/(1-\beta)}$ for $\beta\geq\half$ and $\beta\neq1$, under the constraints $\int dx\,xp(x)=\bar x$ and $\int dx\, p(x)=1$, where integration is now over $x\geq0$. Applying the method of Lagrange multipliers as before, the probability distribution maximising $L_\beta(X)$ is found to have the form
\beq \label{ptilde}
\tilde p(x)=\frac{1}{\tilde\lambda}\tilde p_1(x/\tilde\lambda), \qquad \tilde p_1(x)= \tilde K (1\pm x)^{\frac{-1}{1-\beta}},
\eeq
analogous to Equation~(\ref{pextappb}). The $+$ sign is chosen for $1/2\leq\beta<1$ and the $-$ sign for $\beta>1$, where in the latter case one takes $p_1(x)=0$ for $x>1$. The values of $\tilde\lambda$ and $\tilde K$ follow from the constraints and $1/\alpha+1/\beta=2$ as
\beq
\tilde\lambda = \bar x\, \frac{2\beta-1}{|1-\beta|} = \frac{\bar x}{|1-\alpha|} , \qquad \tilde K= \frac{\beta}{|1-\beta|} = \frac{\alpha}{|1-\alpha|},
\eeq
yielding the upper bound
\beq \label{upperappb}
L_\beta(X) \leq \tilde\lambda L_\beta[\tilde p_1] = \tilde\lambda \tilde K^{-1} (1\pm \bar x/\tilde \lambda)^{\frac{1}{1-\beta}}  = \alpha^{\frac{\alpha}{\alpha-1}} \bar x 
\eeq
for the R\'enyi length of any random variable $X$ on $[0,\infty)$.  Choosing the particular random variable having probability density $p(x)=p(n|\rho)$ for $0\leq n\leq x<n+1$ gives $\bar x=\sum_n p(n|\rho)\int_n^{n+1} dx\,x=\sum_n (n+\half)p(n|\rho)=\langle N\rangle+\half$, and $L_\beta(X)=L_\beta(N|\rho)$, and hence the above upper bound reduces to
\beq \label{upperlappb}
L_\beta(N|\rho) \leq \alpha^{\frac{\alpha}{\alpha-1}} (\langle N\rangle+\half).
\eeq
This also holds in the limit $\alpha\rightarrow1$, for which it gives the strong upper bound $H(N|\rho)\leq \log(\langle N\rangle+\half)e$, which is close to the maximum entropy $\log(\langle N\rangle+1)+\langle N\rangle\log(1+1/\langle N\rangle)$ (to within second order in $1/\langle N\rangle$), corresponding to a thermal state. Finally, substitution of Equation~(\ref{upperlappb}) into Equation~(\ref{rmseappb}) yields the third bound in Theorem~\ref{thm2}, as required.
\end{proof}	

\begin{proof}[Proof of Corollary~\ref{cor1}]
The proof of Equation~(\ref{funappb}) above goes through with $\sigma$ and RMSE replaced by $\Delta_\chi\Phi$, the maximising density $p(x)$ in Equation~(\ref{pextappb}) by $p(x-\chi)$, and $H_\alpha(\Theta_{\rm est}-\Theta|\rho)$ by $H(\Phi|\rho)$. Using uncertainty relation~(\ref{urrenyi}) then leads to Equation~(\ref{cor1a}) of Corollary~\ref{cor1}, in analogy to Equation~(\ref{rmseappb}). The first two bounds in Equation~(\ref{cor1b}) of the corollary correspond to the cases $\alpha\rightarrow\infty$ and $\alpha=\half$, while the third bound follows via the upper bound in Equation~(\ref{upperlappb}) above.
\end{proof}

\begin{proof}[Proof of Equation~(\ref{upperl})]
As noted in the main text, Equation~(\ref{upperl}) for angular momentum is obtained via a suitable modification of the derivation of upper bound~(\ref{upperlappb}) above. This achieved by expanding the domain of integration in the derivation to the full real number line, and replacing the constraint $\int dx\, xp(x)=\bar x$ by $\int dx\, |x|p(x)=\bar x$. This leads to a maximising probability density of the same form as in Equation~(\ref{ptilde}), but extended to the negative numbers and with $x$ replaced by $\bar x$. This increases the corresponding R\'enyi entropy by $\log 2$ and hence  upper bound~(\ref{upperappb}) increases by a factor of 2, to
\beq \label{proofupperl}
L_\beta(X)\leq 2 \alpha^{\frac{\alpha}{\alpha-1}} \bar x .
\eeq
Finally, choosing the random variable having probability density $p(x)=p(j|\rho)$ for $j-\half\leq x<j+\half$, where $p_j=\langle j|\rho|j\rangle$ is the distribution of $J_z$, gives 
\beq
\bar x=\sum_j p_j\int_{j-\half}^{j+\half} dx\,|x| = \frac14 p_0 + \sum_{j\neq 0} p_j |j| = \frac14 p_0 +\langle |J_z|\rangle,
\eeq
and substitution into Equation~(\ref{proofupperl}) gives Equation~(\ref{upperl}) as desired.
\end{proof}

\section{}
\label{appc}

\begin{proof}[Proof of Theorem~\ref{thm3}]
The lower bound in Equation~(\ref{thm3a}) of Theorem~\ref{thm3} was proved in the main text.
The upper bound is obtained by considering a purification $|\psi\rangle\langle \psi|$ of state $\rho$ on the joint Hilbert space of the probe and a reference ancilla $R$.  If $1_R$ denotes the unit operator for the ancilla, then
\beq \label{1ga}
A_{\alpha}^{1_R\otimes G}(|\psi\rangle\langle\psi|) =\inf_{\sigma_{R\E}:[\sigma_{R\E},1\otimes G]=0} D_\alpha(|\psi\rangle\langle\psi|\|\sigma_{R\E}) \geq \inf_{\sigma_{\E}:[\sigma_{\E},G]=0} D_\alpha(\rho\|\sigma_\E) =A^G_\alpha(\rho) ,
\eeq
where the inequality follows by applying data processing inequality~(\ref{data}) to operation of tracing over the ancilla. Further, duality property~(\ref{dual}), proved further below, implies that the left hand side is given by
\begin{align}
A_{\alpha}^{1_R\otimes G}(|\psi\rangle\langle\psi|) &=H_\beta(|\psi\rangle\langle\psi|_{1_R\otimes G}) \nn\\
&= H_\beta\Big(\sum_k(1_R\otimes P_k)|\psi\rangle\langle\psi|(1_R\otimes P_k)\Big)\nn\\
&= \frac{1}{1-\beta} \tr\left[ \left( \sum_k \frac{|\psi_k\rangle\langle\psi_k|}{\langle\psi_k|\psi_k\rangle}  \langle\psi_k|\psi_k\rangle   \right)^\beta \right] \nn\\
&= \frac{1}{1-\beta} \sum_k \langle\psi_k|\psi_k\rangle^\beta,
\label{1gb}
\end{align}
where $|\psi_k\rangle:=(1_R\otimes P_k)|\psi\rangle$,   and the last line follows noting that the $|\psi_k\rangle\langle\psi_k|/\langle\psi_k|\psi_k\rangle$ terms are mutually orthogonal rank-1 projectors. Finally, if $|\psi\rangle=\sum_\lambda\sqrt{w_\lambda} |u_\lambda\rangle_R\otimes |v_\lambda\rangle$ is the Schmidt decomposition of $|\psi\rangle$, so that $\rho=\sum_\lambda w_\lambda |v_\lambda\rangle\langle v_\lambda|$, it follows that $\langle\psi_k|\psi_k\rangle=\sum_\lambda w_\lambda \langle v_\lambda|P_k|v_\lambda\rangle=\tr[\rho P_k]=p(G=g_k|\rho)$, and substituting in Equations~(\ref{1ga}), (\ref{1gb}) gives 
\beq
A^G_\alpha(\rho)\leq A_{\alpha}^{1_R\otimes G}(|\psi\rangle\langle\psi|) = \frac{1}{1-\beta}\sum_k p(G=g_k|\rho)^\beta = H(G|\rho) ,
\eeq
thus yielding the upper bound in Theorem~\ref{thm3}. Finally, if $\rho$ is pure, then duality property~(\ref{dual}) immediately implies this bound is saturated.
\end{proof}

\begin{proof}[Proof of duality property~(\ref{dual})]
The R\'enyi asymmetry of a pure state $|\psi\rangle$ follows from definitions~(\ref{sandwich}) and~(\ref{renyiasymm}) as
\begin{align}
A_\alpha^G(|\psi\rangle\langle\psi|) &= \inf_{\sigma:[\sigma,G]=0} D_\alpha(|\psi\rangle\langle\psi|\|\sigma)\nn\\
&= \inf_{\sigma:[\sigma,G]=0} 
\frac{1}{\alpha-1}\log \tr\left[\left(\frac{\sigma^{\frac{1-\alpha}{2\alpha}}|\psi\rangle\langle\psi|\sigma^{\frac{1-\alpha}{2\alpha}}}{\langle\psi|\sigma^{\frac{1-\alpha}{\alpha}}|\psi\rangle} \, \langle\psi|\sigma^{\frac{1-\alpha}{\alpha}}|\psi\rangle\right)^\alpha\right] \nn\\
&= \inf_{\sigma:[\sigma,G]=0} \frac{1}{\alpha-1}\log \langle\psi|\sigma^{\frac{1-\alpha}{\alpha}}|\psi\rangle^\alpha \nn\\
&= \inf_{\sigma:[\sigma,G]=0} \frac{\alpha}{\alpha-1}\log
\tr\left[ |\psi\rangle\langle\psi|\sum_k P_k\sigma^{\frac{1-\alpha}{\alpha}} P_k\right] \nn\\
&= \inf_{\sigma:[\sigma,G]=0} \frac{\alpha}{\alpha-1}\log
\tr[ |\psi\rangle\langle\psi|_G \,\sigma^{\frac{1-\alpha}{\alpha}} ] ,
\end{align}
where the third line follows because the fractional expression in the round brackets of the second line is a rank-1 projector, the fourth line from $[\sigma,G]=0$, and the last line from the cyclic property of the trace and the definition of $\rho_G$ in Equation~(\ref{asymmbound}).  To determine the infimum over all $\sigma$, note that variation of $\tr[\rho\sigma^{\frac{1-\alpha}{\alpha}}]$ with respect to arbitrary $\sigma$, under the constraint $\tr[\sigma]=1$, yields 
$\sigma=\tilde\sigma:= \rho^{\frac{\alpha}{2\alpha-1}}/\tr[ \rho^{\frac{\alpha}{2\alpha-1}}]$, and the corresponding extremal value $(\tr[ \rho^{\frac{\alpha}{2\alpha-1}}])^{\frac{2\alpha-1}{\alpha}}$. Hence, noting that $[\tilde\sigma,G]=0$ for $[\rho,G]=0$, and that $\beta=\alpha/(2\alpha-1)$ for $1/\alpha+1/\beta=2$, it follows for the choice $\rho=|\psi\rangle\langle\psi|_G$ that
\beq
A_\alpha^G(|\psi\rangle\langle\psi|) = \frac{\alpha}{\alpha-1}\log (\tr[ \rho^{\frac{\alpha}{2\alpha-1}}])^{\frac{2\alpha-1}{\alpha}} = \frac{2\alpha-1}{\alpha-1} \log (\tr[ \rho^{\frac{\alpha}{2\alpha-1}}])^{\frac{2\alpha-1}{\alpha}} = H_\beta(|\psi\rangle\langle\psi|_G),
\eeq
as per Equation~(\ref{dual}).
\end{proof}

\begin{proof}[Proof of Equation~(\ref{infinity}) relating asymmetry to maximally uniform ensembles]
Let $\rho^r_{\E X}$ and $\rho^r_X$ denote the states in Equation~(\ref{holrewrite}) for the prior distribution $p_r(x):=1/(2r)$ for $r<|x|$ and vanishing otherwise. Further, let $\E_r\equiv\{U_x\rho U_x^\dagger;p_r(x)\}$ denote the corresponding continuous ensemble, and $U=\int dx\,U_x^\dagger\otimes|x\rangle\langle x|$ be the controlled unitary transformation used in Section~\ref{sec:asymm}. Fixing $\epsilon>0$ and letting $\sigma_r$ be a state which achieves the infimum in Equation~(\ref{ialpha}) to within less than $\epsilon$, it follows that
\begin{align}
\chi_\alpha(\E_r)+\epsilon &> D_\alpha(\rho^r_{\E X}\|\sigma_r\otimes\rho^r_X) \nn\\
 &= D_\alpha(U\rho^r_{\E X}U^\dagger\|U\sigma_r\otimes\rho^r_X U^\dagger) \nn\\
&= D_\alpha(\rho \otimes\rho^r_X\|\int dx\,p_r(x)U_x^\dagger\sigma_r U_x\otimes |x\rangle\langle x|) \nn\\
&\geq D_\alpha(\rho\|\tilde\sigma_r),
\label{u}
\end{align}
where $\tilde\sigma_r:=\int dr\,p_r(x)U_x^\dagger\sigma_r U_x$ and the last line follows by applying data processing inequality~(\ref{data}) to the partial trace operation.  Hence, recalling  the definition $\chi_\alpha^\infty:=\lim_{r\rightarrow\infty}\chi_\alpha(\E_r)$ in Equation~(\ref{infinity}),
\beq \label{temp}
\chi_\alpha^\infty+\epsilon> D(\rho\|\tilde\sigma_\infty)  .
\eeq
But, using $\sum_k P_k=1$, one has
\begin{align}
\tilde\sigma_\infty&= \lim_{r\rightarrow\infty}\sum_{k,k'} \frac{1}{2r}\int_{-r}^r dx\,e^{ixG} P_k\sigma_r P_{k'}e^{-ixG} = \sum_{k,k'} \lim_{r\rightarrow\infty} \frac{1}{2r}\int_{-r}^r dx\, e^{ix(g_k-g_{k'})}  P_k\sigma_r P_{k'} = \sum_{k} P_k\sigma_\infty P_{k} ,
\end{align}
and thus $[\tilde\sigma_\infty,G]=0$. It then follows from Equation~(\ref{temp}) that
\beq
\chi_\alpha^\infty+\epsilon> \inf_{\sigma:[\sigma,G]=0} D(\rho\|\sigma) = A^G(\rho).
\eeq
Combined with $\chi_\alpha^\infty=\lim_{r\rightarrow\infty}\chi(\E_r)\leq \lim_{r\rightarrow\infty}A_G(\rho)=A_G(\rho)$ from Theorem~\ref{thm3}, this yields
\beq
A_G(\rho)-\epsilon < \chi_\alpha^\infty  \leq A_G(\rho)
\eeq
for all $\epsilon>0$. Hence $A_G(\rho)=\chi_\alpha^\infty$ as per Equation~(\ref{infinity}).
\end{proof}

\section{}
\label{appd}

\begin{proof}[Proof of Theorem~\ref{thm4}]
Defining the translation isometry ${\cal T}:p(a,x)\rightarrow p(x+a,x)$, analogous to the unitary transformation $U$ in Equation~(\ref{u}), and the marginalisation operation ${\cal M}:p(a,x)\rightarrow \int dx\, p(a,x)$, analogous to the partial trace operation, it follows from Equation~(\ref{ialpha}) and data processing inequality~(\ref{data}) for $A=X_{\rm est}$ that
\begin{align}
I_\alpha(X_{\rm est}:X) &= \inf_{q_A} D_\alpha(p_{AX}\| q_A p_X) \nn\\
&= \inf_{q_A} D_\alpha({\cal T}p_{AX}\| {\cal T}q_A p_X) \nn\\
&\geq \inf_{q_A} D_\alpha({\cal MT}p_{AX}\| {\cal MT}q_A p_X) \nn\\
&= \inf_{q} D_\alpha(p_{\rm err}\| q\ast p_-) 
\label{thm4proof}
\end{align}
as per Theorem~\ref{thm4}, where the last line follows using $({\cal MT}p_{AX})(y)=p_{\rm err}(y)$ via Equation~(\ref{perrgen}) and
\beq \label{mt}
({\cal MT}qp_X)(y)=\int dx\, q(y+x)p(x)=\int dx\, q(y-x)p_-(x) = (q\ast p_-)(y).
\eeq
\end{proof}

\begin{proof}[Proof of Corollary~\ref{cor2}]
Define the intervals $I_j:=[(j-\half)|I|,(j+\half)|I|)$ for integer $j$, and the concentration operation ${\cal C}:p(y)\rightarrow p(y\mod I_0)$, that shifts the distribution of $y$ onto interval $I_0$. The probability density $p_{\rm err}$ of $X_{\rm err}=X_{\rm est}-X$ $\tilde p_{\rm err}$ is then related to the probability density $\tilde p_{\rm err}$ of $\tilde X_{\rm err}:=X_{\rm err} \mod I_0$ via $\tilde p_{\rm err}={\cal C}p_{\rm err}$.  More generally,   for a general density $r$  one has by construction that
\beq \label{tilder}
({\cal C}r)(y) = \sum_j r(y+n\ell_I), \qquad y\in I_0.
\eeq
Now, data processing inequality~(\ref{data}) and Equation~(\ref{thm4proof}) above give
\beq \label{dc}
I_\alpha(X_{\rm est}:X) \geq \inf_{q} D_\alpha({\cal C}p_{\rm err}\|{ \cal C}q\ast p_-) = \inf_{q} D_\alpha(\tilde p_{\rm err}\| {\cal C}q\ast p_-) ,
\eeq
where choosing $r=q\ast p_-$ in Equation~(\ref{tilder}) and $p(x)=1/\ell_I$ on some interval $I$ gives
\begin{align}
({\cal C}q\ast p_-)(y) &= \sum_j (q\ast p_-)(y+n\ell_I) \nn\\
&= \frac{1}{\ell_I}\sum_j  \int_I dx\, q(y+x+n\ell_I) \nn\\
&= \frac{1}{\ell_I}\sum_j  \int_{I+n\ell_I} dx\, q(y+x) \nn\\
&= \frac{1}{\ell_I} \int_{-\infty}^\infty dx\, q(y+x) \nn\\
&= \frac{1}{\ell_I} 
\end{align}
(with the second line following via Equation~(\ref{mt})).  Hence Equation~(\ref{dc}) simplifies to
\beq
I_\alpha(X_{\rm est}:X) \geq D_\alpha(\tilde p_{\rm err}\|\ell_I^{-1}) = \log \ell_I- H_\alpha(\tilde p_{\rm err}) = \log \ell_I- H_\alpha(\tilde X_{\rm err}) \geq \log \ell_I- H_\alpha(X_{\rm err})
\eeq
as per  Corollary~\ref{cor2}, with the final inequality following from Lemma~\ref{lem} (proved below).
\end{proof}

\begin{proof}[Proof of Lemma~\ref{lem}]
Entropy is invariant under translations and hence, for the purposes of proving Lemma~\ref{lem}, it is sufficient to consider $\tilde Z=Z \mod I_0=[-\half\ell_I,\half\ell_I)$ for the interval $I_0$ in the proof of Corollary~\ref{cor2} above.  If $q$ denotes the probability density of $Z$, then the corresponding probability density of $\tilde Z$ is then  $\tilde q={\cal C}q$, as per Equation~(\ref{tilder}). 

It is convenient to represent $q$ by a mixture of non-overlapping probability densities, supported on the corresponding non-overlapping intervals $I_j=I_0+n\ell_I$. In particular, one has the identity
\beq
q(z) = \sum_j w_j q_j(z), \qquad w_j:=\int_{I_j} dz\,q(z),\qquad q_j(z):= \left\{ \begin{array}{ll} q(z)/w_j, & z\in I_j, \\
	0, & {\rm otherwise}, \end{array} 
\right.
\eeq
and it follows that $\tilde q=\sum_j w_j \tilde q_j$, where $\tilde q_j={\cal C}q_j$ is supported on $I_0$ and $q_j(z) = \tilde q_j(z+n\ell_I)$. 

For the case of Shannon entropies, i.e.,  $\alpha=1$, it follows that
\begin{align}
H(Z) &= - \sum_j \int_{I_j} dz\,w_jq_j(z)\log w_jq_j(z) = - \sum_j \int_{I_0} dz\,w_j\tilde q_j(z)\log w_j\tilde q_j(z)=H(J)+H(\tilde Z|J),
\end{align}
where $J$ is the discrete random variable with distribution $\{w_j\}$. Hence
\beq
H(Z)-H(\tilde Z) =H(J)+H(\tilde Z|J)-H(\tilde Z)=H(J|\tilde Z)\geq0 ,
\eeq
yielding Lemma~\ref{lem} for the case $\alpha=1$.

Further, for $\alpha\neq1$ the R\'enyi entropies of $Z$ and $\tilde Z$ have the forms
\begin{align}
H_\alpha(Z) &= \frac{1}{1-\alpha}\log \sum_j\int_{I_j} dz\, [w_jq_j(z)]^\alpha
=  \frac{1}{1-\alpha}\log  \int_{I_0} dz\,\sum_j [w_j \tilde q_j(z)]^\alpha , \\
H_\alpha(\tilde Z) &= \frac{1}{1-\alpha}\log \int_{I_0} dz\,\big[\sum_j w_j \tilde q_j(z)\big]^\alpha .
\end{align}
Defining $p_j(z):=w_j\tilde q_j(z)/\sum_j w_j\tilde q_j(z)$, it follows that $\sum_j p_j(z)= 1$, $p_j(z)\leq1$ and $\frac{1}{p_j(z)}\geq1$. Thus, letting $\sum'_j$ denote summation over nonzero values of $p_j(z)$, 
\beq
\frac{\sum_j[w_j\tilde q_j(z)]^\alpha}{[\sum_j w_j \tilde q_j(z)]^\alpha} = {\sum_j}' p_j(z)^\alpha  =
\left\{
\begin{array}{ll} 
	 \sum_j'  p_j(z) \frac{1}{p_j(z)^{1-\alpha}}  \geq  \sum_j'  p_j(z)=1, &\alpha\leq 1,
	\\ 
	 \sum'_j p_j(z) p_j(z)^{\alpha-1}  \leq  \sum_j'  p_j(z)=1, & \alpha\geq 1.
\end{array} 
\right.
\eeq
Hence $\sum_j[w_j\tilde q_j(z)]^\alpha$ is greater than or equal to (less than or equal to) $[\sum_j w_j \tilde q_j(z)]^\alpha$ if $\alpha<1$ ($\alpha>1$), and it immediately follows from the above forms for $H_\alpha(Z)$ and $H_\alpha(Z)$ that
\beq
H_\alpha(Z) \geq H_\alpha(\tilde Z)
\eeq
for the case $\alpha\neq 1$, yielding Lemma~\ref{lem} as desired.
\end{proof}


\begin{thebibliography}{999}
\bibitem{endnote} {In particular, any such discrete partition $\E_{P}\equiv\{p_j;\rho_j \}$ of $\E$,  with associated orthonormal basis $\{|x_j\rangle\}$, can be subpartitioned as $E_{P'}\equiv \{p_{jk};\rho_{jk}\}$ and basis $\{|x_j\rangle\otimes|k\rangle\}$, with $p_j\rho_j=\sum_k p_{jk}\rho_{jk}$. Data processing inequality~(\ref{data}) then gives $\chi_\alpha(\E_P)\leq \chi_\alpha(\E_{P'})$, and so $\chi_\alpha(\E)$ may be defined as the supremum over all such partitions (note this is analogous to Shannon's general definition of mutual information in Appendix~7 of Reference~\cite{Shannon}).
	A possible alternative approach for continuous ensembles would be to work more generally with $C^*$-algebras, where commuting algebras of this type have states formally corresponding to the kets $|x\rangle$, and relative entropies are well-defined and satisfy a data-processing inequality~\cite{Uhlmann}.}
\end{thebibliography}
\end{document}